\documentclass[preprint]{elsarticle}

\usepackage{amsmath, amssymb, amsthm}
\usepackage{mathtools}

\usepackage{enumitem}

\usepackage{graphicx}

\usepackage{xcolor}

\usepackage{glossaries}

\usepackage{makecell}

\usepackage{bm}
\usepackage{algorithmic}
\usepackage{algorithm}
\usepackage[]{subfig}
\usepackage{booktabs}

\usepackage{xr}
\usepackage{hyperref}
\usepackage[capitalise]{cleveref}

\graphicspath{{./img/}{./plots/}}

 \DeclareMathOperator{\E}{\mathsf{E}}
\DeclareMathOperator{\Var}{\mathsf{Var}}
\DeclareMathOperator*{\argmax}{\mathrm{arg\,max}}
\DeclareMathOperator*{\argmin}{\mathrm{arg\,min}}
\DeclareMathOperator{\Tr}{\mathrm{Tr}}

\newcommand{\given}{\,|\,}

\newcommand{\Bgiven}{\,\Big|\,}
\newcommand{\bbgiven}{\,\bigg|\,}

\newcommand{\Hyp}{\mathrm{H}}
\newcommand{\param}{\boldsymbol\theta}
\newcommand{\paramRV}{\boldsymbol\Theta}
\newcommand{\paramScalar}{\theta}
\newcommand{\paramRVScalar}{\Theta}

\newcommand{\paramRVElement}{\Theta}
\newcommand{\seqDataRV}{(\mathbf{X}_n)_{n>0}}
\newcommand{\obsSingle}{\mathbf{x}}
\newcommand{\RVsingle}{\mathbf{X}}
\newcommand{\obsSeq}[1][n]{\mathbf{x}_{1:#1}}

\newcommand{\obsSingleScalar}{x}
\newcommand{\RVsingleScalar}{X}

\newcommand{\natNumbers}{\mathbb{N}}
\newcommand{\natNumbersInclZero}{\natNumbers_{0}}
\newcommand{\reals}{\mathbb{R}}

\newcommand{\Lapl}{\mathrm{Lap}}

\newcommand{\policy}{\pi}
\newcommand{\policySet}{\Pi}
\newcommand{\decR}[1][n]{\delta_{#1}}
\newcommand{\stopR}[1][n]{\Psi_{#1}}

\newcommand{\est}[2][n]{\hat{\boldsymbol\theta}_{#2,#1}}
\newcommand{\estOpt}[2][n]{\hat{\boldsymbol\theta}^\star_{#2,#1}}
\newcommand{\errorCovMatrix}[1][m]{\boldsymbol \Sigma_{#1}}

\newcommand{\detErr}[1][m]{\alpha^{#1}}
\newcommand{\estErr}[1][m]{\beta^{#1}}
\newcommand{\detConstr}{\bar{\alpha}}
\newcommand{\estConstr}{\bar{\beta}}
\newcommand{\detCost}[1][i]{\lambda_{#1}}
\newcommand{\estCost}[1][i]{\mu_{#1}}
\newcommand{\detCostOpt}[1][i]{\lambda^\star_{#1}}
\newcommand{\estCostOpt}[1][i]{\mu^\star_{#1}}

\newcommand{\trueParam}{\param^\star_{\trueHypIdx}}

\newcommand{\trueHyp}{\Hyp_{\trueHypIdx}}
\newcommand{\trueHypIdx}{m^\star}

\newcommand{\iid}{\overset{\,\mathrm{iid}}{\sim}\,}
\newcommand{\dInt}{\mathrm{d}}
\newcommand{\KL}[2]{\mathrm{D}_{\mathrm{KL}}\!\left(\,#1\,\|\,#2\,\right)}

\newcommand{\fishInfMtx}[2][\trueParam]{\bm{\mathcal{I}}_{#2}(#1)}
\newcommand{\fishInfMtxInv}[2][\trueParam]{\bm{\mathcal{I}}^{-1}_{#2}(#1)}

\newcommand{\optRL}{\tau^\star}
\newcommand{\aoRL}{\tau^\circ}
\newcommand{\stopRopt}[1][n]{\Psi_{#1}^\star}
\newcommand{\stopRao}[1][n]{\Psi_{#1}^\circ}

\newcommand{\indd}[1]{\mathbf{1}_{#1}} \newcommand{\ind}[1]{\indd{\{#1\}} }

\newcommand{\Gam}{\mathrm{Gam}}
\newcommand{\invGam}{\mathrm{IGam}}
\newcommand{\unif}{\mathcal{U}}
\newcommand{\norm}[2]{\ensuremath{\mathcal{N}\left(#1,#2\right)}}
\newcommand{\complexGauss}{\mathcal{CN}}

\newcommand{\gDiv}{\bar{g}}
\newcommand{\DDiv}[1][m]{\bar{D}_{#1,n}}
\newcommand{\detCostDiv}[1][i]{\bar{\lambda}_{#1}}
\newcommand{\estCostDiv}[1][i]{\bar{\mu}_{#1}}
\newcommand{\divCost}{\max\{\detCost[1],\ldots,\detCost[M],\estCost[1],\ldots,\estCost[M]\}}

\newcommand{\gradEstMC}{\hat\nabla_{\estCost[]}}
\newcommand{\gradDetMC}{\hat\nabla_{\detCost[]}}

\newcommand{\detErrMC}[1][m]{\hat\alpha^{#1}}
\newcommand{\estErrMC}[1][m]{\hat\beta^{#1}}
\newcommand{\smallNonNegNumber}{\varepsilon}

\newcommand{\invHess}{\mathbf{H}}

\newcommand{\detTol}[1][m]{\varepsilon_{\detErr[]}^{#1}}
\newcommand{\estTol}[1][m]{\varepsilon_{\estErr[]}^{#1}}

\newcommand{\llr}{\eta}

\newcommand{\QAMsymbol}[1][m]{A_{#1}}
\newcommand{\noiseSingle}{v}

\newcommand{\noisePower}{\sigma^2}

\newcommand{\etal}{\emph{et al.}{} }

\newlength{\imgWidthSingle}
\newlength{\imgWidthDouble}
\newtheorem{definition}{Definition}
\newtheorem{theorem}{Theorem}
\newtheorem{lemma}{Lemma}

\newacronym{apo}{APO}{asymptotically pointwise optimal}
\newacronym{ao}{AO}{asymptotically optimal}
\newacronym{awgn}{AWGN}{additive white Gaussian noise}
\newacronym{iid}{iid}{independent and identically distributed}
\newacronym{kl}{KL divergence}{Kullback-Leibler divergence}
\newacronym{MAE}{MAE}{mean absolute error}
\newacronym{mmse}{MMSE}{minimum mean-squared error}
\newacronym{MSE}{MSE}{mean-squared error}
\newacronym{msprt}{MSPRT}{Matrix Sequential Probability Ratio Test}
\newacronym{pdf}{pdf}{probability density function}
\newacronym{sprt}{SPRT}{sequential probability ratio test}
\newacronym{qam}{QAM}{quadrature amplitude modulation}
\newacronym{bfgs}{BFGS}{Broyden–Fletcher–Goldfarb–Shanno}
\newacronym{fim}{FIM}{Fisher's information matrix}
 \makeglossaries

\crefname{assumption}{Assumption}{Assumptions}
\newlist{assenum}{enumerate}{1} \setlist[assenum]{label={\textbf{A\arabic*}}}
\crefalias{assenumi}{assumption}

\crefname{equation}{}{}

\begin{document}
\title{Asymptotically Optimal Procedures for Sequential Joint Detection and Estimation}
\renewcommand{\floatpagefraction}{.8}

\newcommand{\new}[1]{\textcolor{blue}{#1}}
\newcommand{\IEEEQED}{\mbox{\rule[0pt]{1.3ex}{1.3ex}}}
\maketitle

\begin{frontmatter}

\author[1]{Dominik Reinhard\corref{cor1}}
\ead{reinhard@spg.tu-darmstadt.de}

\author[2]{Michael Fau\ss{}
\fnref{fn1}'}
\ead{mfauss@princeton.edu}

\author[1]{Abdelhak M. Zoubir}
\ead{zoubir@spg.tu-darmstadt.de}

\cortext[cor1]{Corresponding author}

\fntext[fn1]{The work of Michael Fau\ss{} was supported by the German Research Foundation (DFG) under grant number 424522268.}

\affiliation[1]{organization={Technische Universität Darmstadt},
postcode={64283},
city={Darmstadt},
country={Germany}}

\affiliation[2]{organization={Department of Electrical Engineering, Princeton University},
postcode={NJ 08544},
city={Princeton},
country={USA}}

\begin{abstract}
 We investigate the problem of jointly testing multiple hypotheses and estimating a random parameter of the underlying distribution in a sequential setup.
 The aim is to jointly infer the true hypothesis and the true parameter while using on average as few samples as possible and keeping the detection and estimation errors below predefined levels.
 Based on mild assumptions on the underlying model, we propose an asymptotically optimal procedure, i.e., a procedure that becomes optimal when the tolerated detection and estimation error levels tend to zero.
 The implementation of the resulting asymptotically optimal stopping rule is computationally cheap and, hence, applicable for high-dimensional data.
 We further propose a projected quasi-Newton method to optimally choose the coefficients that parameterize the instantaneous cost function such that the constraints are fulfilled with equality.
 The proposed theory is validated by numerical examples.
\end{abstract}

\end{frontmatter}

\glsresetall
\section{Introduction}
There exist various scenarios in which hypothesis testing and parameter estimation occur in a coupled way and the outcome of both is important.
That is, one would like to decide among several hypotheses and, depending on the decision, estimate some parameter of the model under the selected hypothesis.
In radar, for example, one would like to detect the presence of a target and, if a target is present, to estimate, e.g., its position or velocity \cite{tajer2010Optimal, chen2018impact}.
To perform dynamic spectrum access in cognitive radio, the secondary user has to detect the primary user and to estimate the possible interference\cite{Yilmaz2014Sequential}.
Detecting the presence of a signal and estimating the channel is of interest in many communication scenarios\cite{jan2018iterative}.
There exist many more applications in which such scenarios arise as, for example, biomedical engineering\cite{chaari2013fast,Makni2008}, speech processing\cite{Momeni2015Joint}, visual inference\cite{Vo2010Joint} or changepoint detection\cite{boutoille2010hybrid}.

The problem of treating hypothesis testing and parameter estimation jointly dates back to the late 1960s.
Middleton and Espositio used a Bayesian framework to find a jointly optimal solution \cite{Middleton1968Simultaneous}.
That framework was later extended to multiple hypotheses \cite{fredriksen1972simultaneous}.
After a period of declining interest, this problem has regained more attention since 2000\cite{tajer2010Optimal,Yilmaz2014Sequential,jan2018iterative,chaari2013fast,Makni2008,Momeni2015Joint,Vo2010Joint,boutoille2010hybrid}.

In many applications, one would like to perform inference by using as few samples as possible and maintaining its quality at the same time.
These requirements can, for example, be caused by constraints on the latency or the power consumption of a system.
In the late 1940s, Abraham Wald introduced the field of sequential analysis with the invention of the \gls{sprt}\cite{wald1947sequential}.
In sequential inference, data is collected until one is confident enough about the phenomenon of interest.
One advantage of sequential methods is that they use on average signifacantly fewer samples compared to their counterparts with a fixed sample size while having the same inference quality.
An overview of sequential detection and sequential estimation methods is, for example, given in \cite{tartakovsky2014sequential} and \cite{ghosh2011sequential}, respectively.

Treating the problem of joint detection and estimation using a fixed number of samples allows to trade-off detection and estimation accuracy.
Contrary to this, the additional degree of freedom, i.e., the number of used samples, in a sequential setup allows to control the detection and estimation errors individually.
This is important because detection and estimation often have contradicting requirements on the data.
Hence, using a two-step procedure, e.g., a sequential detector followed by an optimal estimator, would allow to control the detection errors, but most probably results in very high estimation errors, which is not desired.
This is illustrated by an example.
As mentioned before, detecting the presence of a target and estimating its position is a typical problem in radar.
In this application, the detection performance as well as the estimation accuracy are important.
If the radar system observes a peak with a very high amplitude that points into the direction of the target, its presence can be detected reliably, but its position cannot be estimated accurately if the peak is too broad.
Conversely, if the system observes a very narrow peak of a small amplitude, the position of the target can be estimated accurately, but, on the other hand, its presence cannot be detected reliably.
Therefore, the procedure has to gather more samples until it observes a very narrow peak of a high amplitude, i.e., until the target can be detected reliably and its position can be estimated accurately.
Note that this is an oversimplification to highlight the underlying fundamental \emph{joint} detection and estimation problem.
\subsection{State-of-the-Art}
There is only little prior work on joint detection and parameter estimation in a sequential framework.
Although there exist sequential hypothesis tests that include an estimation step, like the generalized \gls{sprt} or the adaptive \gls{sprt}, estimation only occurs as an intermediate step and the primary interest lies in the outcome of the test.
More details on the generalized \gls{sprt} and the adaptive \gls{sprt} can be found in \cite[Section 5.4]{tartakovsky2014sequential}.

Simultaneous detection and estimation in a sequential setting was first mentioned in \cite{birdsall1973sufficient}.
However, that work considers only sequential updating of a sufficient statistic for the problem of joint detection and estimation rather than presenting a solution of the joint inference problem.

Sequential joint detection and state estimation of Markov signals was investigated in \cite{grossi2008sequential,buzzi2006joint}.
However, although estimating the true state is of primary interest, these works do not consider the uncertainty about the true state in the number of used samples.
A sequential procedure that overcomes this problem was proposed for a related problem in \cite{grossi2009sequential}.

Initially, Y\i lmaz \etal investigated the problem of optimal sequential joint detection and estimation \cite{yilmaz2015sequential}.
In that work, a binary hypothesis test was first performed and, in case a decision was made in favor of the alternative, a random parameter is to be estimated.
The method in \cite{yilmaz2015sequential} was extended to multiple hypotheses \cite{Yilmaz2016Sequential} and applied to joint spectrum sensing and channel estimation \cite{Yilmaz2014Sequential}.
In the works by Y\i lmaz \emph{et al.}, a weighted sum of detection and estimation errors is used as a cost function for the joint inference problem.
This cost function is kept below a predefined level while minimizing the number of used samples for \emph{every} set of observations.

In \cite{reinhard2019bayesian}, we proposed a novel framework for optimal sequential procedures for sequential joint detection and estimation.
In that framework, the aim is to minimize the \emph{expected} number of samples while keeping the detection and estimation errors below predefined levels.
That framework was later extended to multiple hypotheses \cite{reinhard2020multiple,reinhard2021multiple} and applied to joint signal detection and signal-to-noise ratio estimation \cite{reinhard2019JointSNR} as well as to joint symbol decoding and noise power estimation\cite{reinhard2020multiple,reinhard2021multiple}.
Sequential joint detection and estimation in distributed sensor networks was investigated in \cite{reinhard2020Distributed,reinhard2020DistributedNonGaussian}.

Although the procedures in \cite{reinhard2019bayesian, reinhard2021multiple} are strictly optimal, i.e., there is no other procedure that uses less samples on average and fulfills the constraints, their design can be computationally costly as its implementation requires a discretization of the state space.
To overcome this, we presented an asymptotically pointwise optimal procedure \cite{reinhard2021apo}, i.e., a procedure that becomes optimal when the tolerated detection and estimation error levels become sufficiently small.
However, the work in \cite{reinhard2021apo} relies on strict assumptions. Namely, the hypotheses only differ in the prior of the random parameter and the data conditioned on the random parameter is \gls{iid} with a distribution from the exponential family.

\subsection{Contributions and Outline}
In this work, we propose an \gls{ao} procedure for the problem of sequential joint detection and estimation under mild assumptions on the underlying model.
To implement the optimal procedures \cite{reinhard2019bayesian, reinhard2021multiple}, one has to find a sufficient statistic that represents the data (see \cite[Assumption A3]{reinhard2021multiple}), discretize its state space and numerically solve a set of recursively defined Bellman equations after restricting the maximum number of samples.
This has the following drawbacks: first, it is not possible to find a sufficient statistic for every problem (see \cref{sec:num_res_unknown_noise}). Second, the state space of the sufficient statistic has to be of a low dimension such that it can be discretized accurately.
Next, the optimal procedure has to be truncated, which can be challenging as the number of required samples has to be known prior to designing the procedure.
Finally, solving the recursively defined Bellman equations is computationally demanding.

The proposed AO procedure overcomes these challenges as it does not require a restriction on the maximum number of samples and it is easy to implement because it does not require solving the Bellman equations, despite achieving similar performance to the strictly optimal procedure \cite{reinhard2019bayesian, reinhard2021multiple}.
Moreover, it can deal with multiple parameters under each hypothesis, whereas the optimal procedure is restricted to single scalar parameters under each hypothesis.

The contributions of this manuscript are as follows.
First, we present an asymptotically optimal stopping rule for sequential joint detection and estimation that depends on a set of parameters.
Next, we show that the gradients of the \gls{ao} cost function are asymptotically proportional to the corresponding error metric, i.e., they become proportional as the nominal detection and estimation error levels tend to zero.
Based on this connection, it is then shown that the search for the optimal coefficients, i.e., the coefficients for which the procedure hits the nominal error levels, can be formulated as a convex optimization problem.
Finally, a projected quasi-Newton method to solve this problem is proposed.

This work is structured as follows.
In \cref{sec:problemFormulation}, the notation and the underlying assumptions are presented.
Moreover, some fundamentals of sequential joint detection and estimation are revisited, the performance measures are introduced and a detailed problem formulation is given.
For the sake of completeness and to introduce the relevant quantities, the reduction of the original problem to an optimal stopping problem is sketched in \cref{sec:reductionOptimalStopping}.
In \cref{sec:ao}, the proposed \gls{ao} policy is introduced and rigorous proofs of optimality are given.
The optimal choice of the coefficients parameterizing the \gls{ao} policy along a projected quasi-Newton algorithm is presented in \cref{sec:optCostCoeff}.
To validate the proposed theory, numerical examples are given in \cref{sec:numResults}.
Conclusions are drawn in \cref{sec:conclusions}.
 
\section{Problem Formulation}\label{sec:problemFormulation}
Let $\mathcal{X} = \seqDataRV$ be a sequence of random variables that is generated under one out of $M$ different hypotheses $\Hyp_m$, $m\in\{1,\ldots,M\}$.
Under each hypothesis, the distribution of the data depends on a real random parameter $\paramRV_m\in\reals^{K_m}$, $K_m\geq 1$, with known distribution.
The realization of $\paramRV_m$ is denoted by $\param_m$.
However, there may exist scenarios, in which the likelihood under the different hypotheses do not share the same parameter.
To cover these scenarios and have a precise notation, we decided to add the index $m$ to the parameter.
The occurrence of the hypotheses is also random with known prior probabilities $p(\Hyp_m)$, $m\in\{1,\ldots,M\}$.
The sequence of random variables is assumed to be conditionally \gls{iid}, i.e., the random variables $\mathcal{X}\given\Hyp_m,\param_m$ are \gls{iid}.
Hence, one can express the $M$ different hypotheses as
\begin{equation}\label{eq:hyps}
    \Hyp_m: \; \RVsingle_n\given\param_m \iid p(\obsSingle\given\Hyp_m,\param_m)\,,\; \paramRV_m\sim p(\param_m\given\Hyp_m)\,,
\end{equation}
where $m\in\{1,\ldots,M\}$.
Generally speaking, the true hypothesis may affect both, the distribution of the data given the parameter as well as the prior distribution on the parameters.
Since the hypotheses are composite ones, we are, besides deciding in favor of the true hypothesis, also interested in inferring the parameter that generates the sequence $\mathcal{X}$.
Using an optimal procedure for the hypothesis test and an optimal estimator would not necessarily result in an overall optimal performance\cite{moustakides2012joint}.
Therefore, the problem of detection and estimation has to be considered \emph{jointly}.
Moreover, the sequence $\mathcal{X}$ is observed sample by sample and one stops as soon as one is confident enough about the true hypothesis and the true parameter.
Hence, the problem is one of \emph{sequential joint detection and estimation}.

Before a more technical problem formulation is provided, the notation and assumptions are summarized and some fundamentals on sequential joint detection and estimation are introduced.
\subsection{Notations and Assumptions}\label{sec:assumptions}
Unless otherwise stated, random variables and their realizations are denoted by upper-case and lower-case symbols, respectively.
To keep the notation simple, the integration domain is not indicated explicitly for integrals that are taken over the entire domain, e.g.,
$\E[\RVsingle] = \int_{-\infty}^\infty \obsSingle p(\obsSingle)\dInt\obsSingle \equiv \int \obsSingle p(\obsSingle)\dInt\obsSingle\,$.
Moreover, the dependency of functions, such as estimators, on the observations is dropped for the sake of a compact notation.
The indicator function of event $\mathcal{A}$ is denoted by $\indd{\mathcal{A}}(x)$, i.e., $\indd{\mathcal{A}}(x) = 1$ iff $x\in\mathcal{A}$.
To end up with a compact notation, the argument of the indicator function is dropped, i.e., $\E[\indd{\mathcal{A}}(\RVsingle)]$ is replaced by $\E[\indd{\mathcal{A}}]$.
The trace of a matrix is $\Tr(\cdot)$ and $\obsSeq=\obsSingle_1,\ldots,\obsSingle_n$.
The \gls{kl}, \gls{fim} under $\Hyp_m$ and its inverse are denoted by $\KL{p(\obsSingle)}{q(\obsSingle)}$, $\fishInfMtx[\param_m]{m}$ and $\fishInfMtxInv[\param_m]{m}$, respectively.
Given a random sequence $a_n$ and a deterministic variable $a$, almost sure convergence $a_n\xrightarrow{a.s.} a$ is defined as $P(\lim_{n\rightarrow\infty} a_n=a)=1$. The limit $n\rightarrow\infty$ is dropped for the sake of a compact notation.
For the remaining convergence statements, the limits are stated explicitly.

 The following assumptions are made throughout the paper:
\begin{assenum}
 \item The true parameter and the true hypothesis stay constant during the observation period.
\item \label{ass:finiteSecondOrder} The variances of $\paramRV_m=[\paramRVElement^1_m,\ldots,\paramRVElement^{K_m}_m]$ are positive and finite, i.e.,
        \begin{align*}
         0 < \Var[\paramRVElement^k_m\given\Hyp_m] < \infty\,, \quad & \forall m\in\{1,\ldots,M\}\,,\;\forall k\in\{1,\ldots,K_m\}\,.
        \end{align*}

 \item \label{ass:fishInf} For \gls{fim} under hypothesis $\Hyp_m$, it holds that
    \begin{align*}
        P(0<\Tr(\fishInfMtxInv[\param_m]{m})<\infty)=1\,,\quad & \forall m\in\{1,\ldots,M\}\,.
    \end{align*}
  \item\label{ass:posKL} For all $m\in\{1,\ldots,M\}$, all $k\in\{1,\ldots,M\}\setminus m$, all $\param_m$, $\param_k$ it holds that
 \begin{align*}
  0<\KL{ p(\obsSingle\given\Hyp_m,\param_m)}{ p(\obsSingle\given\Hyp_k,\param_k) } < \infty\,.
 \end{align*}
 \item\label{ass:localLipschitz} Let $\trueParam$ denote the true parameter and let $U$ denote an open neighborhood of $\trueParam$. Then, there exist a square integrable function $f_{\trueParam}(\obsSingle)$ such that for all $\param, \param^\prime\in U$ it holds that
 \begin{align*}
  \biggl\lvert \log \frac{ p(\obsSingle\given\trueHyp,\param)}{ p(\obsSingle\given\trueHyp,\param^\prime)}\biggr\rvert \leq & f_{\trueParam}(\obsSingle) \|\param-\param^\prime\| \quad  P(\cdot\given\trueHyp,\trueParam)- \text{a.s.}\,.
 \end{align*}
 In other words, $\log p(\mathbf{x}\,|\,\mathrm{H}_{m^\star},\boldsymbol\theta)$ is locally Lipschitz continuous over $U$.
\item\label{ass:klTaylor} The \gls{kl} has a second order Taylor-expansion around $\trueParam$, i.e.,
    \begin{align*}
\E\left[ \log\frac{p(X\given\trueHyp,\trueParam)}{p(X\given\trueHyp,\param_{\trueHypIdx})}\bbgiven \trueHyp, \trueParam\right]  & =  \frac{1}{2}(\trueParam-\param_{\trueHypIdx})^\top \mathbf{A}(\trueParam-\param_{\trueHypIdx}) \\
     & \phantom{{}={}}+ o(\|\trueParam-\param_{\trueHypIdx}\|^2)\,,
    \end{align*}
    where $\mathbf{A}$ is a positive definite matrix.
 
\end{assenum}
\cref{ass:posKL} guarantees the proper convergence of $p(\Hyp_m\,|\, \param_m, \obsSeq)$, whereas \cref{ass:localLipschitz,ass:klTaylor} guarantee the proper convergence of $p(\param_m\given\Hyp_m,\obsSeq)$.
The validity of the assumptions is theoretically shown in some cases (see supplement \cite{reinhard2022Supplement}).
\subsection{Fundamentals of Sequential Joint Detection and Estimation}
To solve a joint detection and estimation problem, one has to find a decision rule as well as a set of estimators, i.e., one estimator under each hypothesis.
The decision rule $\decR[n]\in\{1,\ldots,M\}$ maps the observations to a decision in favor of a particular hypothesis and the estimators $\est{m}$ map the observations to a point in the state space of a particular parameter.
As indicated by the subscript, the decision rule and the estimators depend on the number of samples $n$, which also determines the dimensionality of their arguments.
In a sequential framework, the number of samples is not given \emph{a priori}, but depends on the observations themselves.
More precisely, one would stop sampling as soon as the certainty about the true hypothesis and the true parameter is strong enough.
Therefore, a stopping rule $\stopR[n]\in\{0,1\}$, that maps the observations to a decision whether to stop or continue sampling, has to be found.
The run-length, i.e., the number of used samples, can hence be defined as
\begin{equation*}
 \tau = \min\{n\geq 0: \stopR(\obsSeq)=1\}\,.
\end{equation*}
Since the run-length depends on the random data, the run-length is also a random variable.

In what follows, the set of the stopping rule, the decision rule and the estimators is referred to as policy, and is defined as $\pi = \{\stopR, \decR, \est{1},\ldots,\est{M}\}_{n\in\natNumbersInclZero}$.
and the set of all feasible policies is denoted by $\policySet$.

In this work, we use similarly to \cite{reinhard2019bayesian, reinhard2021multiple}, the probability of falsely rejecting a hypothesis and the \gls{MSE} for quantifying the detection and estimation errors, respectively.
The former is defined as
$
\detErr = \E\Bigl[\ind{\decR[\tau]\neq m}\Bgiven\Hyp_m\Bigr]\,,
$
with $\{\decR \neq m\} := \{\obsSeq: \decR(\obsSeq) \neq m\}$.
For the latter, we set the estimation error to zero in case of a wrong decision, which can be written as
\begin{equation*}
 \estErr = \E\Bigl[\ind{\decR[\tau]= m}\bigl\|\est[\tau]{m} - \paramRV_m\bigr\|^2\Bgiven\Hyp_m\Bigr]\,.
\end{equation*}

\subsection{Sequential Joint Detection and Estimation as an Optimization Problem}
As mentioned before, the aim is to design a sequential scheme that uses on average as few samples as possible while keeping the detection and estimation errors below predefined levels.
More formally, the design problem can be formulated as the following constrained optimization problem
\begin{align}\label{eq:constrProblem}
\begin{split}
    \min_{\policy\in\policySet}\; \E[\tau]\,, \quad
    \text{s.t.} \quad & \detErr \leq \detConstr_m\,, \quad\estErr \leq \estConstr_m\,, \quad  m\in\{1,\ldots,M\}\,, \\
 \end{split}
\end{align}
where $\detConstr_m\in(0,1)$ and $\estConstr_m\in(0,\infty)$ are the maximum tolerated detection and estimation errors, respectively.

Deriving and implementing the estimators and the decision rule that solve \cref{eq:constrProblem} is quite straightforward, whereas it becomes challenging for the corresponding stopping rule.
Therefore, the constrained problem in \cref{eq:constrProblem} is solved by the following steps.
First, it is reduced to an optimal stopping problem in \cref{sec:reductionOptimalStopping}.
Next, an asymptotically optimal stopping rule, which depends on a set of parameters, is derived in \cref{sec:ao}.
Finally, in \cref{sec:optCostCoeff}, a systematic way how to choose these parameters such that all constraints of \cref{eq:constrProblem} are fulfilled with equality is presented.
Note that although the constraints are fulfilled with equality, the proposed method does not minimize the expected number of samples exactly.
However, as shown in \cref{sec:numResults}, the proposed method requires only slightly more samples on average than its optimal counterpart, i.e., the performance is close to optimal.

\subsection{Reduction to an Optimal Stopping Problem}\label{sec:reductionOptimalStopping}
In order to reduce \cref{eq:constrProblem} to an optimal stopping problem, it is first converted to an unconstrained problem, whose objective is a weighted sum of average run-length, detection and estimation errors.
This unconstrained problem acts as an auxiliary problem.
Next, the unconstrained optimization problem is solved with respect to the decision rule and estimators to end up with an optimal stopping problem.
Although the steps are already presented in \cite{reinhard2019bayesian, reinhard2021multiple}, this section is added for the sake of completeness and for introducing the quantities that are required for the subsequent sections.

The unconstrained version of \cref{eq:constrProblem} is given by
\begin{align}\label{eq:unconstrProblem}
    \min_{\policy\in\policySet}\;\left\{\E[\tau] + \sum_{m=1}^M p(\Hyp_m)(\detCost[m]\detErr + \estCost[m]\estErr)\right\}\,,
\end{align}
where $\detCost[m],\estCost[m]>0, m\in\{1,\ldots,M\}$, are some cost coefficients, which are assumed to be fixed for now.
For suitably chosen $\detCost[m],\estCost[m], m\in\{1,\ldots,M\}$, the solution of \cref{eq:unconstrProblem} also solves \cref{eq:constrProblem}.
Problem \cref{eq:unconstrProblem} is first minimized with respect to the decision rule and then with respect to the estimators.
The optimal estimators and the optimal decision rule are given by
\begin{align*}
 \estOpt{m}(\obsSeq) = \E[\paramRV_m\given\Hyp_m,\obsSeq]\,,\,\quad\text{and}\quad
 \decR^\star(\obsSeq) = \argmin_{m\in\{1,\ldots,M\}} D_{m,n}(\obsSeq)\,,
\end{align*}
where $D_{m,n}(\obsSeq)$ denotes the cost for deciding in favor of $\Hyp_m$ at time $n$.
Let the error covariance matrix be defined as
\begin{align*}
\errorCovMatrix(\obsSeq) = \E[(\paramRV_m-\estOpt{m}(\obsSeq ))(\paramRV_m-\estOpt{m}(\obsSeq ))^\top ]\,.
\end{align*}
Then, the cost $D_{m,n}(\obsSeq)$ is defined as
\begin{align}\label{eq:defD}
D_{m,n}(\obsSeq) &= \sum_{i=1, i\neq m}^M \detCost p(\Hyp_i\given\obsSeq)  +\estCost[m] p(\Hyp_m\given\obsSeq)\Tr(\errorCovMatrix(\obsSeq))\,.
\end{align}
Finally, we end up with the following optimization problem
\begin{align}\label{eq:optStopProb}
 \min_{\{\stopR\}_{n\geq 0}}\; \E\left[\tau+g(\obsSeq[\tau])\right]\,,
\end{align}
with the cost function
\begin{align}\label{eq:defG}
 g(\obsSeq) = \min\{D_{1,n}(\obsSeq),\ldots,D_{M,n}(\obsSeq)\}\,.
\end{align}

Solving \cref{eq:optStopProb} is usually a challenging task in the design of optimal sequential procedures.
In \cite{reinhard2019bayesian, reinhard2021multiple}, for example, the problem is truncated, i.e., the maximum number of samples is restricted to a fixed constant $N$, and the truncated optimal stopping problem is solved via backward induction.

In this work, however, we do not seek for an optimal solution, but rather for stopping rules that are optimal in an asymptotic sense.
The concept of \gls{ao} policies is introduced in the next section.
  \section{Asymptotically Optimal Procedures}\label{sec:ao}
Before it can be shown how to define an \gls{ao} stopping rule in the context of this work, some fundamentals on \gls{ao} stopping rules have to be revisited.
For this, we consider a different model that is widely used in the context of (Bayesian) sequential inference \cite{bickel1967asymptotically,bickel1968asymptotically,ghosh2011sequential}.
Let
\begin{align}\label{eq:optStopModel}
 Z_n(c) = & nc + Y_n
\end{align}
be a sequence of random variables, where $c>0$ is the cost for observing one sample and $Y_n$ is some random variable.
The aim is to find an optimal stopping time for model \cref{eq:optStopModel}, i.e., a stopping time $\tau$ that minimizes $\E[Z_\tau(c)]$.
Deriving a strictly optimal stopping time, and, hence, a strictly optimal stopping rule that generates this stopping time, can be computationally demanding.
If the maximum number of samples is restricted, the optimal stopping rule can as, e.g., in \cite{reinhard2019bayesian,reinhard2021multiple}, be calculated via dynamic programming, which becomes computationally infeasible with increasing dimensionality.
For the solution of high-dimensional optimal stopping problems, see, for example, \cite{becker2019deep} and references therein.
To overcome this, asymptotically optimal stopping rules are considered here.
That is, stopping rules that become optimal when the cost per sample tends to zero, i.e., $c\rightarrow0$.
According to \cite{bickel1968asymptotically}, a stopping rule is called \glsfirst{ao} if there exists no other feasible stopping rule that results \emph{on average} in smaller costs when the cost per sample tends to zero.
The formal definition is stated below.
\begin{definition}[Bickel and Yahav \cite{bickel1968asymptotically}] \label{def:AO_bickel}
 Consider the model in \cref{eq:optStopModel}. Then, a stopping time $t$ is asymptotically optimal if
 \begin{align*}
  \frac{\E[Z_{t(c)}(c)]}{\inf_{s\in T} \E[Z_s(c)]} \xrightarrow{a.s.} 1\quad \text{as}\; c\rightarrow 0\,,
\end{align*}
 where $T$ is the set of all feasible stopping times.
\end{definition}

Contrary to model \cref{eq:optStopModel}, in which the cost per observation $c$ is the free parameter, the model in \cref{eq:optStopProb} has a fixed unit cost per observation.
Instead, the cost function $g(\obsSeq)$ in \cref{eq:optStopProb} is parameterized by the coefficients $\detCost[m], \estCost[m]\geq0$, $m\in\{1,\ldots,M\}$.
As outlined in \cite{reinhard2019bayesian,reinhard2021multiple}, these coefficients can be chosen systematically such that the resulting policy also solves \cref{eq:constrProblem}.
Since the nominal detection and estimation error levels are the free parameters during the design process, it is more intuitive to formulate the asymptotic regime in terms of these levels rather than in terms of the cost coefficients.
That is, the term asymptotically refers to the case when the levels of all constraints tend to zero, which is equivalent to $\max\{\detConstr_1,\ldots,\detConstr_M,\estConstr_1,\ldots,\estConstr_M\}\rightarrow 0$.
Loosely speaking, tightening a constraint implies increasing the corresponding coefficients.
In \cite[Appendix F]{reinhard2019bayesian}, it is shown that some constraints may be implicitly fulfilled by other constraints and, thus the corresponding coefficient becomes zero.
However, in this work, it is assumed that all coefficients are strictly positive.
Nevertheless, the coefficients can be chosen arbitrarily small and, hence, the optimal policy can be approximated arbitrarily well.

In this context, a stopping time $t$ is called \gls{ao} if
\begin{align}\label{eq:AO_sjde}
\frac{\E[t + g(\obsSeq[t])]}{\inf_{s\in T} \E[s + g(\obsSeq[s])]} \xrightarrow{a.s.} 1
\qquad\text{as}\quad \max\{\detConstr_1,\ldots,\detConstr_M, \estConstr_1,\ldots,\estConstr_M\}\; \rightarrow 0\,,
 \end{align}where $T$ is the set of all feasible stopping times.

Before the \gls{ao} stopping rule is derived, a variant of \cref{eq:optStopProb} is introduced.
Let $\bar{c} = \left(\divCost\right)^{-1}$ and 
\begin{align}\label{eq:defModG}
 \gDiv(\obsSeq) = \left(\divCost\right)^{-1}g(\obsSeq)\,,
\end{align}
then \cref{eq:optStopProb} is equivalent to finding an optimal stopping rule for
\begin{align}\label{eq:auxOptStopProblem}
 \bar{c} n + \gDiv(\obsSeq)\,.
\end{align}
As \cref{eq:auxOptStopProblem} and \cref{eq:optStopProb} differ only up to a constant, finite, and non-zero scaling factor, they share the same minimizer.
The reason for considering this problem will become clear in the sequel.

In what follows, $\trueHyp$ and $\trueParam$ denote the true hypothesis and the true parameter, respectively.
Moreover, the short-hand notations
\begin{align*}
  \detCostDiv[m] & = \frac{\detCost[m]}{\divCost}\,, \quad \estCostDiv[m] &= \frac{\estCost[m]}{\divCost}\,,
 \end{align*}
for $m\in\{1,\ldots,M\}$ are used.
As a first step, it has to be shown that the cost function $\gDiv(\obsSeq)$ tends to zero almost surely as the number of observations tends to infinity.
Moreover, it has to be shown that $n\gDiv(\obsSeq)$ converges to a finite and non-zero random variable.
This is stated in the following theorem.
\begin{theorem}\label{theo:prop}
 Let $\gDiv(\obsSeq)$ be as defined in \cref{eq:defModG}.
 Then, when $\detCost[m],\estCost[m]>0$, $m\in\{1,\ldots,M\}$, and the assumptions stated in \cref{sec:assumptions} are fulfilled, it holds that
 \begin{enumerate}
  \item $P(\gDiv(\obsSeq)>0) = 1$
  \item $\gDiv(\obsSeq) \xrightarrow{a.s.} 0 $
  \item $n\gDiv(\obsSeq)\xrightarrow{a.s.} G = \estCostDiv[\trueHypIdx]\Tr(\fishInfMtxInv{\trueHypIdx})$.
  It further holds that $P(0<G<\infty)=1$.
 \end{enumerate}
\end{theorem}

Before we can prove \cref{theo:prop}, two auxiliary results about the convergence of the posterior probability $p(\Hyp_m\given\obsSeq)$ and the error covariance marix $\errorCovMatrix$ are introduced.

\begin{lemma}\label{lem:convPostH}
 For $\seqDataRV$ be defined as above, it holds that 
 \begin{align*}
  p(\Hyp_m\given \obsSeq) \xrightarrow{a.s.}
    \begin{cases}
        1 & m = \trueHypIdx \\
        0 & \text{else}\,.
    \end{cases}
 \end{align*}
 Moreover, the posterior distribution converges exponentially.
\end{lemma}
See \cref{ext:proof:convPostH} in \cite{reinhard2022Supplement} for a proof.

\begin{lemma}\label{lem:convPostVar}
  For $\seqDataRV$ be defined as above, it holds that 
\begin{align*}
  n\Tr(\errorCovMatrix[\trueHypIdx]) & \xrightarrow{a.s.}
        \Tr(\fishInfMtxInv{\trueHypIdx})\,,\quad
  \Tr(\errorCovMatrix[\trueHypIdx]) \xrightarrow{a.s.} 0\,.
 \end{align*}
\end{lemma}
See \cref{ext:proof:convPostVar} in \cite{reinhard2022Supplement} for a proof.

\begin{proof}[Proof of \cref{theo:prop}]
    First, the statement that $\gDiv(\obsSeq)$ is positive with probability one, is proven.
    The function $\gDiv(\obsSeq)$ is defined as
    \begin{align*}
     \gDiv(\obsSeq) = \min\{\DDiv[1](\obsSeq), \ldots, \DDiv[M](\obsSeq)\}\,,
    \end{align*}
    where $\DDiv[m](\obsSeq) =  \sum_{i=1, i\neq m}^M \detCostDiv p(\Hyp_i\given\obsSeq) +\estCostDiv[m] p(\Hyp_m\given\obsSeq)\Tr(\errorCovMatrix)$.
    The function $\DDiv[m](\obsSeq)$ is a linear combination of the posterior probabilities {$p(\Hyp_i\given\obsSeq)$}, $i\in\{1,\ldots,M\}\setminus m$ and of the product $p(\Hyp_m\given\obsSeq)\Tr(\errorCovMatrix)$.
    The posterior probabilities as well as the posterior variances, i.e., the diagonal elements of the matrix $\errorCovMatrix$, are positive.
    It further holds that $\estCostDiv[m]$ is positive.
    Hence, it follows that $\DDiv(\obsSeq)$ is positive with probability one and, since $\gDiv(\obsSeq)$ is the minimum of all
    $\DDiv(\obsSeq)$, $m\in\{1,\ldots,M\}$, $\gDiv(\obsSeq)$ is also positive with probability one.
    
    To prove Statement 2, we assume that the sequence $\obsSeq$ is generated under hypothesis $\trueHyp$ and consider the limit of the sequence $\DDiv[\trueHypIdx](\obsSeq)$ as $n\rightarrow\infty$.
    According to \cref{lem:convPostH}, $p(\Hyp_m\given\obsSeq)\xrightarrow{a.s.}0$ for all $m\in\{1,\ldots,M\}\setminus\trueHypIdx$ and, hence,
    \begin{align}\label{eq:convDetPenZero}
     \sum_{i=1, \,i\neq \trueHypIdx}^M \detCostDiv p(\Hyp_i\given\obsSeq)\xrightarrow{a.s.}0\,.
    \end{align}
    Next, the product $p(\trueHyp\given\obsSeq)\Tr(\errorCovMatrix[\trueHypIdx])$ needs closer inspection.
    According to \cref{lem:convPostH}, it holds that $p(\trueHyp\given\obsSeq)\xrightarrow{a.s.}1$.
    In line with \cref{lem:convPostH} and \cref{lem:convPostVar}, it further holds that
    \begin{align*}
     \Bigl(p(\trueHyp\given\obsSeq),\Tr(\errorCovMatrix[\trueHypIdx])\Bigr) \xrightarrow{a.s.} \Bigl(1,0\Bigr)\,.
    \end{align*}
    Since the logarithm is a continuous function we can, by applying the continuous mapping theorem \cite{mann1943stochastic,taboga2017}, state that
    \begin{align*}
     \log p(\trueHyp\given\obsSeq) + \log \Tr(\errorCovMatrix[\trueHypIdx]) \xrightarrow{a.s.} -\infty\,.
    \end{align*}
    Applying, again, the continuous mapping theorem yields
    \begin{align*}
     p(\trueHyp\given\obsSeq) \Tr(\errorCovMatrix[\trueHypIdx]) \xrightarrow{a.s.} 0\,.
    \end{align*}
    Thus, we can conclude that
$\DDiv[\trueHypIdx](\obsSeq) \xrightarrow{a.s.} 0 $\,,
which implies $\gDiv(\obsSeq)\xrightarrow{a.s.}0$.
    
    It is left to prove Statement 3, i.e., that $n\gDiv(\obsSeq)$ converges almost surely to a random variable $G$ that is positive and finite with probability one.
    Again, let $\trueHyp$ and $\trueParam$ denote respectively the hypothesis and the parameter under which the sequence $\obsSeq$ is generated.
    \cref{lem:convPostH} states that $p(\Hyp_m\given\obsSeq)\xrightarrow{a.s.}0$ for all $m\neq\trueHypIdx$ and that the posterior probabilities converge at an exponential rate.
    Therefore, it follows that
    \begin{align*}
     \sum_{i=1, i\neq \trueHypIdx}^M n \detCostDiv p(\Hyp_i\given\obsSeq)\xrightarrow{a.s.}0\,.
    \end{align*}
    According to \cref{lem:convPostH} and \cref{lem:convPostVar}, it holds that
    \begin{align*}
     \Bigl(p(\trueHyp\given\obsSeq),n\Tr(\errorCovMatrix[\trueHypIdx])&\Bigr)  \xrightarrow{a.s.} \biggl(1,\Tr(\fishInfMtxInv{\trueHypIdx})\biggr)\,.
    \end{align*}
    From the continuous mapping theorem, we conclude that
    \begin{align*}
     \log p(\trueHyp& \given\obsSeq) + \log\bigl(n\Tr(\errorCovMatrix[\trueHypIdx])\bigr)  \xrightarrow{a.s.} \log 1 + \log \biggl(\Tr(\fishInfMtxInv{\trueHypIdx})\biggr)
    \end{align*}
    and applying the exponential function results in
    \begin{align*}
     p(\trueHyp\given\obsSeq)n\Tr(\errorCovMatrix[\trueHypIdx]) & \xrightarrow{a.s.} \Tr(\fishInfMtxInv{\trueHypIdx})\,.
    \end{align*}
    Hence, we conclude that
    \begin{align*}
     n\DDiv[\trueHypIdx](\obsSeq) \xrightarrow{a.s.} \estCostDiv[\trueHypIdx]\Tr(\fishInfMtxInv{\trueHypIdx})\,.
    \end{align*}
    For $m\neq\trueHypIdx$, it follows from \cref{lem:convPostH} that
    \begin{align*}
     \sum_{i=1, i\neq m}^M n \detCostDiv p(\Hyp_i\given\obsSeq)\xrightarrow{a.s.}\infty
    \end{align*}
    and therefore
$n\DDiv[m](\obsSeq) \xrightarrow{a.s.} \infty\,.$
According to the definition of $\gDiv(\obsSeq)$, it holds that
    \begin{align*}
     n\gDiv(\obsSeq) \xrightarrow{a.s.} G = \estCostDiv[\trueHypIdx]\Tr(\fishInfMtxInv{\trueHypIdx})\,.
    \end{align*}
    It is now left to show that $G$ is positive and finite with probability one.
    According to \cref{ass:fishInf}, the trace of the inverse of \gls{fim} is positive and finite with probability one.
    Moreover, as $\estCostDiv[m]$, $m\in\{1,\ldots,M\}$, are positive, we can conclude that
    $P(0<G<\infty)=1$.
\end{proof}
In order to propose an \gls{ao} stopping rule for the problem of sequential joint detection and estimation, it has to be shown that the expected value of $\gDiv(\obsSeq)$ exists and is finite for all $n>0$.
This is stated in the following lemma.
\begin{lemma}\label{lem:finiteEx}
    Let $\gDiv(\obsSeq)$ be as defined above, then, under the assumptions stated in \cref{sec:assumptions}, it holds that
    \begin{align*}
        \E[\gDiv(\obsSeq)] < \infty\,,\quad \forall n\geq1\,.
    \end{align*}
\end{lemma}
A proof is given in \cite[\cref{ext:proof:finiteEx}]{{reinhard2022Supplement}}.
Based on these properties, one can show that the stopping rule
\begin{align}\label{eq:stopRule}
 \text{stop as soon as}\quad \frac{g(\obsSeq)}{n+1}\leq1
\end{align}
is \gls{ao}.
This is stated in the following theorem.
\begin{theorem}\label{theo:AO}
 Assume that all coefficients $\detCost[m],\estCost[m]$, $m\in\{1,\ldots,M\}$ are positive. Then, the stopping rule \cref{eq:stopRule} is \gls{ao} in the sense of \cref{eq:AO_sjde}.
\end{theorem}
\begin{proof}
    To prove this theorem, it is first shown that \begin{align}\label{eq:stopRule_gen}
        \text{stop as soon as}\quad \frac{\gDiv(\obsSeq)}{n+1}\leq\bar{c}
    \end{align}
    is \gls{ao} in the sense of \cref{def:AO_bickel}.
    In order for \cref{eq:stopRule_gen} to be \gls{ao}, the conditions in \cite[Theorem 2.1.]{bickel1968asymptotically} and \cite[Theorem 3.1.]{bickel1968asymptotically} have to be fulfilled.
    As $\detCost[m],\estCost[m]>0$, $m\in\{1,\ldots,M\}$, the conditions of \cite[Theorem 2.1.]{bickel1968asymptotically} are fulfilled by \cref{theo:prop}.
    Hence, it is left to show that the condition stated in \cite[Theorem 3.1.]{bickel1968asymptotically}, i.e.,
    \begin{align}\label{eq:sup_exp}
     \sup n \E[\gDiv(\obsSeq)] <\infty\,,
    \end{align}
    is fulfilled.
    According to \cref{lem:finiteEx}, the expectation of $\gDiv(\obsSeq)$ is finite.
    This implies that \cref{eq:sup_exp} holds as long as $n$ is finite.
    In \cref{theo:prop}, it is shown that $ng(\obsSeq)$ converges to a random variable $G$ that is finite with probability one, which implies that $n\E[\gDiv(\obsSeq)]<\infty$ for $n\rightarrow \infty$.
    Therefore, \cref{eq:sup_exp} is true and \cref{eq:stopRule_gen} is \gls{ao} in the sense of \cref{def:AO_bickel} if the cost coefficients $\detCostDiv[m], \estCostDiv[m]$ are finite for all $m\in\{1,\ldots,M\}$.
Since $\detCostDiv[m]$ and $\estCostDiv[m]$ stay positive and finite when the tolerated detection and estimation error levels tend to zero, the stopping rule \cref{eq:stopRule_gen} is \gls{ao} in the sense of \cref{def:AO_bickel}.
    
        It is left to show that $\detConstr^m,\estConstr^m\rightarrow0$ implies that $\detCost[m],\estCost[m]\rightarrow\infty$.
When using the optimal decision rule and the optimal estimators, $\detErr,\estErr\rightarrow0$ can only be achieved if the number of samples tends to infinity, i.e., $\tau\rightarrow\infty$, since we assume non-zero and finite Fisher's information and KL-divergences according to \cref{ass:fishInf,ass:posKL}.
When using the proposed stopping rule stated in \cref{eq:stopRule}, an infinite run-length can be achieved if and only if the cost for stopping tends to infinity, i.e., $g(\obsSeq)\rightarrow\infty$ for all $\obsSeq$.
From the definition
\begin{align*}
 g(\obsSeq) = \min\{D_{1,n}(\obsSeq),\ldots,D_{M,n}(\obsSeq)\}\,,
\end{align*}
one can see that $g(\obsSeq) \rightarrow \infty$ for all $\obsSeq$ holds if and only if $D_{m,n}(\obsSeq)\rightarrow\infty$ for all $m\in\{1,\ldots,M\}$ and all $\obsSeq$.
From \cref{eq:defD}, one can see that $D_{m,n}(\obsSeq)\rightarrow\infty$ implies that $\detCost[m],\estCost[m]\rightarrow\infty$, since the posterior probabilities are finite by definition and $\Tr(\errorCovMatrix(\obsSeq))$ is finite according to \cref{ass:fishInf}.
Hence, a procedure only meets the nominal error levels $\detConstr^m,\estConstr^m\rightarrow0$ if $\detCost[m],\estCost[m]\rightarrow\infty$.
    By the definition of $\bar{c}$, it follows that $\bar{c}\rightarrow 0$ if the nominal error levels tend to zero.
    Hence, the stopping rule \cref{eq:stopRule} is \gls{ao} in the sense of \cref{eq:AO_sjde}.
\end{proof}

 \section{Optimal Choice of the Cost Coefficients}\label{sec:optCostCoeff}
In the previous section, it was assumed that all coefficients tend to infinity such that the proposed procedure is \gls{ao}.
However, the ultimate goal is to choose the coefficients such that all constraints in \cref{eq:constrProblem} are fulfilled and, assuming that the resulting coefficients are sufficiently large, the resulting procedure is close to optimal.
This section addresses the problem of how to select the coefficients optimally, i.e., such that all constraints are fulfilled.

In \cite{reinhard2019bayesian,reinhard2021multiple}, we have shown that there is a strong connection between the derivatives of the optimal cost function with respect to these coefficients and the performance measures.
However, it is not obvious that such a relation exists also for the \gls{ao} stopping rule proposed in this work.
In what follows, let $\stopRopt$ denote the optimal stopping rule presented in \cite{reinhard2021multiple} and let $\optRL$ denote the corresponding stopping time.
The stopping rule and stopping time of the \gls{ao} procedure are denoted by $\stopRao$ and $\aoRL$, respectively.

The connection between the derivatives of the \gls{ao} solution of the optimal stopping problem with respect to the coefficients and the performance measures is stated in the following theorem.
\begin{theorem}\label{theo:derivatives}
 Let $\stopRao$ denote the \gls{ao} stopping rule defined in \cref{eq:stopRule} and $\aoRL$ the corresponding stopping time.
 Then, it holds for $\max\{\detConstr_1,\ldots,\detConstr_M, \estConstr_1,\ldots,\estConstr_M\}\; \rightarrow 0$ that
\begin{align*}
  \frac{\partial}{\partial\detCost[m]} \E[\aoRL + g(\obsSeq[\aoRL])] \rightarrow p(\Hyp_m)\detErr\,,\quad
  \frac{\partial}{\partial\estCost[m]} \E[\aoRL + g(\obsSeq[\aoRL])] \rightarrow p(\Hyp_m)\estErr\,.
 \end{align*}
\end{theorem}
The proof of \cref{theo:derivatives} is outlined in \cite[\cref{ext:app:proofDerivatives}]{reinhard2022Supplement}.
Based on the result stated in \cref{theo:derivatives}, we can proceed as in \cite{reinhard2019bayesian, reinhard2021multiple} and obtain the optimal cost coefficients via 
\begin{align}\label{eq:finalMaxProblem}
 \max_{\detCost[],\estCost[]>0}\;\E[\aoRL + g(\obsSeq[\aoRL])] - \sum_{m=1}^M p(\Hyp_m)(\detCost[m]\detConstr_m + \estCost[m]\estConstr_m)\,.
\end{align}
As it is not trivial to see that strong duality between \cref{eq:constrProblem} and \cref{eq:finalMaxProblem} holds asymptotically, it is fixed in the next theorem.
\begin{theorem}\label{theo:asymptotic_duality}
Let $\detCostOpt[]$, $\estCostOpt[]$ denote the solution of \cref{eq:finalMaxProblem}.
Then for the \gls{ao} policy parameterized by these coefficients it holds that
\begin{align*}
 \detErr & = \detConstr_m,\,\quad
 \estErr = \estConstr_m,\,\quad m\in\{1,\ldots,M\}
\end{align*}
and
\begin{align*}
 \E[\aoRL] \rightarrow \E[\optRL]\quad\text{as}\quad \max\{\detConstr_1,\ldots,\detConstr_M, \estConstr_1,\ldots,\estConstr_M\}\; \rightarrow 0\,.
\end{align*}
That is, a solution of \cref{eq:finalMaxProblem} also solves \cref{eq:constrProblem} asymptotically.
\end{theorem}
The proof of \cref{theo:asymptotic_duality} can be found in \cite[\cref{ext:proof:asymptotic_duality}]{reinhard2022Supplement}.

In order to solve \cref{eq:finalMaxProblem}, we propose a projected quasi-Newton method as summarized in \cref{alg:projQuasiNewton}.
\begin{algorithm}[t]
    \caption{Projected Quasi-Newton Method}
    \label{alg:projQuasiNewton}
 \begin{algorithmic}[1]
\STATE \textbf{input} \par
\hskip\algorithmicindent Nominal levels $\detConstr_m, \estConstr_m$, $m\in\{1,\ldots,M\}$ and initial coefficients $\detCost[]^{(0)}, \estCost[]^{(0)}$
\STATE \textbf{initialization}:
$k \leftarrow 0$
        \REPEAT
\STATE Perform Monte Carlo simulation and obtain gradients via \cref{eq:gradDetMC} and \cref{eq:gradEstMC}\vspace{2mm}
        \STATE  $ \hat\nabla^{(k)} \leftarrow - 
                                                \left[
\gradDetMC^{(k)} \quad 
                                                \gradEstMC^{(k)} 
\right]^\top$
            \vspace{2mm}
            \IF {$k=0$}
            \STATE  $\invHess^{(k)} = \frac{1}{\lVert \hat\nabla^{(k)} \rVert}\mathbf{I}$
            \ELSE
            \vspace{2mm}
            \STATE $\mathbf{s} \leftarrow \left[ 
{\detCost[]}^{(k)} \quad
                                                {\estCost[]}^{(k)}  
\right]^\top
                                                -
                                    \left[ 
{\detCost[]}^{(k-1)} \quad
                                                {\estCost[]}^{(k-1)}  
\right]^\top
                                                $
            \vspace{2mm}
            \STATE $ \mathbf{y} \leftarrow \hat\nabla^{(k)} - \hat\nabla^{(k-1)}$
            \IF {$k=1$}
                \vspace{2mm}
                \STATE $\invHess^{(k-1)}\leftarrow \frac{\mathbf{y}^\top \mathbf{s}}{\mathbf{y}^\top \mathbf{y}}\mathbf{I}$
                \vspace{2mm}
            \ENDIF
            \vspace{2mm}
            \STATE $\invHess^{(k)} = \bigl(\mathbf{I} - \frac{\mathbf{s} \mathbf{y}^\top}{\mathbf{y}^\top s} \bigr) \invHess^{(k-1)} \bigl(\mathbf{I} - \frac{\mathbf{y} \mathbf{s}^\top}{\mathbf{y}^\top s} \bigr) + \frac{\mathbf{s}\mathbf{s}^\top}{\mathbf{y}^\top \mathbf{s}}$
            \vspace{2mm}
            \ENDIF
            \vspace{2mm}
\STATE $\left[
                                                {\detCost[]}^{(k+1)}\quad
                                                {\estCost[]}^{(k+1)}
                                                \right]^\top
                                                = \max\left\{
                                    \left[ {\detCost[]}^{(k)}\quad
                                                {\estCost[]}^{(k)}
\right]^\top - \invHess^{(k)} \hat\nabla^{(k)},\; \smallNonNegNumber \right\}$
        \vspace{2mm}
        \STATE $\gamma\leftarrow 1$
        \WHILE{\cref{eq:full_condition} fulfilled}
                \STATE $\gamma\leftarrow\gamma_\text{red}\gamma$
        \ENDWHILE
        \STATE $k\leftarrow k+1$
        \UNTIL convergence
\end{algorithmic}
 \end{algorithm}
The gradients of \cref{eq:finalMaxProblem} with respect to $\detCost[]$ and $\estCost[]$, which depend on the detection and estimation error levels, are given in \cite[Eq. \cref{ext:eq:grad}]{reinhard2022Supplement}.
As the \gls{ao} procedure is not truncated, calculating the gradients directly is not possible.
Therefore, the gradients have to be estimated via Monte Carlo simulations.
Let $\detErrMC$ and $\estErrMC$ denote the Monte Carlo estimate of $\detErr$ and $\estErr$, respectively.
Then, the estimates of the gradients become
\begin{align}
 \label{eq:gradDetMC}\gradDetMC & = [\,p(\Hyp_1)(\detErrMC[1]-\detConstr_1)\,,\,\ldots\,,\,p(\Hyp_M)(\detErrMC[M]-\detConstr_M)\,]^\top\,,\\
 \label{eq:gradEstMC}\gradEstMC & = [\,p(\Hyp_1)(\estErrMC[1]-\estConstr_1)\,,\,\ldots\,,\,p(\Hyp_M)(\estErrMC[M]-\estConstr_M)\,]^\top\,. 
\end{align}
Although calculating the gradients via Monte Carlo simulations might seem computationally demanding, they are also required for evaluating the objective in \cref{eq:finalMaxProblem}.
Hence, estimating the gradients comes at no additional computational costs compared to the evaluation of the objective.
The estimates are then used to update the inverse of the Hessian matrix $\invHess$ via the \gls{bfgs} rule \cite[Eq. (6.17)]{nocedal2006numerical}.
To get the cost coefficients at the current iteration, the ones from the previous iteration are updated via the step $\invHess^{(k)}\hat\nabla^{(k)}$.
Choosing the step size is an important aspect in many optimization algorithms.
Usually one either performs an exact or an inexact line search.
In the former, the objective is minimized in the search direction, whereas in the latter the objective is only approximately minimized in the search direction or a sufficient decrease is found.
To evaluate the objective or the gradients, a Monte Carlo simulation must be performed. This is the most computationally costly part of the optimization algorithm. Therefore, the number of objective and gradient evaluations should be kept to a minimum.
For the problem at hand, we resort to a backtracking line search.
Commonly used conditions for accepting a step like Wolfe's or  Armijo's conditions turned out to be too conservative for the problem at hand.
This means that the aforementioned rules reject a step too often, resulting in many Monte Carlo simulations being required.
In this work, a custom step size selection rule is applied that is tailored to the problem at hand.
Let
\begin{align*}
    f\left(\left[\detCost[] \quad \estCost[]\right]^\top\right) = -\left(\widehat{\E[\tau]} + \sum_{m=1}^M p(\Hyp_m) \detCost[m](\detErrMC-\detConstr_m) + \sum_{m=1}^M p(\Hyp_m) \estCost[m](\estErrMC-\estConstr_m)\right)
\end{align*}
denote the objective when a policy parametrized by $\detCost[]$ and $\estCost[]$ is used.
Then, a candidate step size $\gamma$ is accepted if the objective decreases, i.e.,
\begin{align}\label{eq:decrease}
 f\left(\left[ {\detCost[]}^{(k+1)}\quad
                                                {\estCost[]}^{(k+1)}
\right]^\top\right) \leq f\left(\left[ {\detCost[]}^{(k)}\quad
                                                {\estCost[]}^{(k)}
\right]^\top\right)\,.
\end{align}
A rejected step size is then reduced according to
\begin{align}\label{eq:step_size_red}
    \gamma\leftarrow\gamma_\text{red}\gamma
\end{align}
 with $\gamma_\text{red}\in(0,1)$.
 Prior to this, the step size is reduced according to \cref{eq:step_size_red} as long as the length of the resulting step is above a maximum length $\gamma_\text{max}$, i.e.,
 \begin{align}
  \label{eq:step_size_max}
  \|\gamma\invHess^{(k)}\hat\nabla^{(k)}\| > \gamma_\text{max}\,.
 \end{align}
Hence, no Monte Carlo simulations, which is the most costly part of this procedure, is performed for too large steps.
Additionally, to avoid non-meaningful procedures, i.e., procedures that always reject a particular hypothesis, the step is further rejected if
 \begin{align}\label{eq:max_det_err_one}
    \max\{\detErrMC[1],\ldots,\detErrMC[M]\} = 1
 \end{align}
  even if this would recude the final objective.
  Finally, to prevent an infinite loop in the step size reduction, it is only applied if the length of the step is sufficiently large, i.e.,
 \begin{align}\label{eq:step_size_red_stop}
    \|\gamma\invHess^{(k)}\hat\nabla^{(k)}\| > \gamma_\text{min}\,.
 \end{align}
 The condition for rejecting a step is if
 \begin{align}
    \label{eq:full_condition}
  \cref{eq:step_size_max} \;\text{or}\; \cref{eq:max_det_err_one}\;\text{or}\;\cref{eq:step_size_red_stop}
 \end{align}
are fulfilled.
This rule for step size selection worked well for the examples presented in this work.
Finally, to ensure that the coefficients are strictly positive, they are projected onto the set $[\smallNonNegNumber,\infty)$, where $\smallNonNegNumber$ is a small positive number, e.g., $\smallNonNegNumber=10^{-12}$.
These steps are repeated until convergence, i.e., until
\begin{align}
\label{eq:stopCrit}
 \lvert \detErrMC - \detConstr_m \rvert & \leq \detTol
 \quad \text{and} \quad 
 \lvert \estErrMC - \estConstr_m \rvert \leq \estTol
\end{align}
hold for all $m\in\{1,\ldots,M\}$.
The positive tolerances $\detTol$ and $\estTol$ have to be set by the designer in advance.
The stopping criterion only depends on the estimated error levels $\detErrMC, \estErrMC$, which are anyway required in any step, as well as the nominal error levels $\detConstr_m,\estConstr_m$ and the tolerances $\detTol,\estTol$.
Therefore, \eqref{eq:stopCrit} can be evaluated straightforwardly without the need for further calculations.
 \section{Numerical Results}\label{sec:numResults}
In this section, we apply the proposed theory to three problems.
The first one is simple and is used to show the basic properties of the resulting policy and to compare it to the optimal policy from \cite{reinhard2021multiple}.
The remaining examples are used to show the applicability of the proposed method to more complex problems.
If not otherwise stated, the parameters outlined in \cite[Appendix A]{reinhard2022Supplement} are used for the design of the \gls{ao} procedure.

\subsection{Benchmarking Methods}
There exist only little related work in the field of sequential joint detection and estimation and the optimal procedure proposed in \cite{reinhard2021multiple} is only applicable in case a low-dimensional representation of the data exists.
Therefore, we use two-step procedures, one sequential and one using a fixed number of samples, for benchmarking.
First a sequential two-step procedure as in, e.g., \cite{reinhard2021multiple}, is used as a reference for comparison.
This method is referred to as S-TS.
That is, we use a sequential detector followed by an \gls{mmse} estimator.
Although there exist different sequential tests for multiple hypotheses \cite{baum1994sequential,tartakovsky2014sequential}, we resort to the \gls{msprt}\cite[Section 4.1]{tartakovsky2014sequential} due to its easy implementation and nice asymptotic properties.
The \gls{msprt} uses the pairwise log-likelihood ratios between $\Hyp_m$ and $\Hyp_j$ for $m\neq j$
\begin{align*}
 \llr_{mj}(\obsSeq) = \log\Biggl(\frac{p(\obsSeq\given\Hyp_m)}{p(\obsSeq\given\Hyp_j)}\Biggr)\,,\; &m,j\in\{1,\ldots,M\}\,.
\end{align*}
The stopping rule and the decision rule of the \gls{msprt} are given by \cite[Eqs. (4.3) and (4.4)]{tartakovsky2014sequential}
\begin{align*}
\stopR^\text{MSPRT} & = \begin{cases}
                         1 & \exists m: \llr_{mj}\geq A_{mj}\,,\; \forall j\in\{1,\ldots,M\}\setminus m\,,\\
                         0 & \text{else}\,,
                        \end{cases}\\
\decR^\text{MSPRT} & = \biggl\{m: \llr_{mj}\geq A_{mj}\,,\; \forall j\in\{1,\ldots,M\}\setminus m \biggr\}\,.
\end{align*}
The stopping rule and the decision rule are parameterized by a set of parameters $A_{mj}, m,j\in\{1,\ldots,M\}$.
In order to keep the error probabilities $P(\{\decR[\tau]\neq m\}\given\Hyp_m)$ below $\detConstr_m$ for all $m\in\{1,\ldots,M\}$, the thresholds can be chosen as \cite[Eq. (4.9)]{tartakovsky2014sequential}
$
 A_{mj} = A_m \approx \log(M-1) - \log(\detConstr_m)\,.
$
The S-TS procecure is neither optimal nor asymptotically optimal for the problem of sequential joint detection and estimation as the stopping rule is determined by the \gls{msprt} which does not take any estimation error into account.
However, the \gls{mmse} estimator is the optimal estimator in the \gls{MSE} sense and, under certain conditions, the \gls{msprt} is asymptotically optimal.
Hence, the S-TS procedure uses an asymptotically optimal detector followed by an optimal estimator.
Nevertheless, the performance is not overall optimal as will be demonstrated by the examples in this section.
This emphasizes the benefits of a joint design.

Second, a two-step procedure using a fixed amount of data, which is referred to as FSS-TS, is used for comparison.
The FSS-TS consists of a fixed-sample size maximum \emph{a posterior} detector, followed by an MMSE estimator.
The decision rule of the former is given by
\begin{align*}
  \delta^\text{FSS}(\obsSeq[n]) = \argmax_m\; p(\Hyp_m\given\obsSeq[n])\,.
\end{align*}
To select the number of samples $n$ for the FSS-TS, the average run-length of the asymptotically optimal procedure is rounded towards the next largest integer.
This procedure, which consists of an optimal detector in the Bayesian sense as well as an optimal estimator, is used to illustrate the performance gap when the number of samples is not adapted to the data.

\subsection{Shift-in-Mean Scenario}
The first example is used to show the basic properties of the proposed policy and to highlight the differences between the \gls{ao} policy and the optimal policy.
The same scenario was already used in \cite{reinhard2021multiple} for the optimal policy.
In that scenario, the aim is to decide among three different hypotheses with equal prior probability.
Under all hypotheses, the likelihood, i.e., the distribution of the data conditioned on the random parameter, is a Gaussian distribution with known variance $\sigma^2$.
Hence, the data distribution depends on a single scalar random parameter whose prior distribution differs under the three hypotheses.
The three hypotheses differ only in the prior distribution of the mean. The priors have a disjoint support.
More formally, the hypotheses are\begin{align*}
  \begin{split}
   \Hyp_1:&\; \RVsingleScalar_n\given\paramScalar_1 \iid \norm{\mu_1}{\sigma^2}\,, -\mu_1 + 1.3 \sim \Gam(1.7,1)\,,\\ \Hyp_2:&\; \RVsingleScalar_n\given\paramScalar_2 \iid \norm{\mu_2}{\sigma^2}\,, \phantom{- + 1.3}\mu_2 \sim \unif(-1,1)\,,\\ \Hyp_3:&\; \RVsingleScalar_n\given\paramScalar_3 \iid \norm{\mu_3}{\sigma^2}\,, \phantom{-}\mu_3 - 1.3 \sim \Gam(1.7,1)\,, \end{split}
\end{align*}
where $\norm{\paramScalar}{\sigma^2}$ is the normal distribution with mean $\paramScalar$ and variance $\sigma^2$, $\unif(l,u)$ is the uniform distribution on the interval $[l,u)$ and $\Gam(a,b)$ is the Gamma distribution with shape and scale parameters $a$ and $b$, respectively.
The variance of the Gaussian distribution is set to $\sigma^2=4$.

The aim is to design a sequential scheme that jointly infers the true hypothesis and the true parameter.
The detection errors should be limited to $\detConstr_1=\detConstr_2=\detConstr_3=0.05$ and the estimation errors should be limited to $\estConstr_1=0.2, \estConstr_2=0.15, \estConstr_3=0.1$.

In addition to the two-step procedures, we use the optimal scheme presented in \cite{reinhard2021multiple} for benchmarking purposes.
As discussed before, optimal policies are only known for truncated schemes. In this example, the truncation length is set to $100$ samples.
However, there are no restrictions on the number of samples for the \gls{ao} scheme.
Moreover, in order to design the optimal procedure and to visualize the optimal and the \gls{ao} policy, a sufficient statistic in the sense of \cite[Assumption A3]{reinhard2021multiple} has to be found.
The sufficient statistic, which serves as a low-dimensional representation of the data, is set to $\bar{x}_n = n^{-1}\sum_{i=1}^n x_i$ and captures all relevant information about the true hypothesis, the true parameter and the next sample, for the shift-in-mean scenario. See \cite[Section 6.2]{reinhard2021multiple} for details.
We refer to \cite[Section 6.2]{reinhard2021multiple} for a description of the remaining parameters used during the design process of the optimal scheme.

The optimal parameters of the \gls{ao} policy are obtained via \cref{alg:projQuasiNewton}, where $10^6$ Monte Carlo runs are used to estimate the gradients.
Moreover, the initial set of cost coefficients and the tolerances are set to $\detCost[m]^{(0)}=\estCost[m]^{(0)}=100$ for all $m\in\{1,\ldots,3\}$ and $\detTol=\estTol=0.005$ for all $m\in\{1,\ldots,3\}$, respectively.
Before we present the results of the \gls{ao} policy, we show that the assumptions are fulfilled for the problem at hand. The proof that \crefrange{ass:fishInf}{ass:klTaylor} hold is laid down in the supplement \cite[\cref{ext:app:sim_assumptions}]{reinhard2022Supplement}.
As the moments of $\paramRVScalar_m\given\Hyp_m$ are finite at least up to order two, the proposed policy is \gls{ao}.

\begin{table}[t]
    \centering
    \caption{Shift-in-mean scenario: simulation results of the asymptotically optimal (AO), optimal (opt) and two-step (TS) procedure.}
    \label{tbl:shiftInMean}
    \subfloat[Detection and estimation errors.\label{tbl:shiftInMean_err}]{\begin{tabular}{@{}rccccc@{}}
	\toprule
& \makecell{nominal \\level} & AO & opt & S-TS & FSS-TS\\
	\midrule

	$\detErr[1]$ & $0.050$ & $0.050$ & $0.050$ & $0.021$ & $0.042$\\
	$\detErr[2]$ & $0.050$ & $0.048$ & $0.050$ & $0.052$ & $0.075$\\
	$\detErr[3]$ & $0.050$ & $0.047$ & $0.051$ & $0.021$ & $0.042$ \\
	\midrule
	$\estErr[1]$ & $0.200$ & $0.199$ & $0.199$ & $0.809$ & $0.151$ \\
	$\estErr[2]$ & $0.150$ & $0.154$ & $0.149$ & $0.181$ & $0.113$\\
	$\estErr[3]$ & $0.100$ & $0.100$ & $0.099$ & $0.805$ & $0.150$ \\

	\bottomrule
\end{tabular}
 }\hfil
    \subfloat[Average run-lengths.\label{tbl:shiftInMean_rl}]{\begin{tabular}{@{}rcccc@{}}
	\toprule
& AO & opt & S-TS & FSS-TS \\
	\midrule

	$\E[\tau\given\Hyp_1]$ & $15.28$ & $15.07$ & $10.06$ & $21.00$ \\
	$\E[\tau\given\Hyp_2]$ & $14.78$ & $14.73$ & $19.21$ & $21.00$ \\
	$\E[\tau\given\Hyp_3]$ & $32.23$ & $31.76$ & $10.19$ & $21.00$ \\
	$\E[\tau]$ & $20.76$ & $20.51$ & $13.16$ & $21.00$ \\

	\bottomrule
\end{tabular}
 }
\end{table}

\begin{figure}[t]
    \centering
    \includegraphics{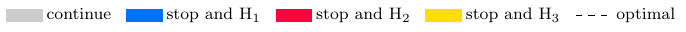}\\
    \includegraphics{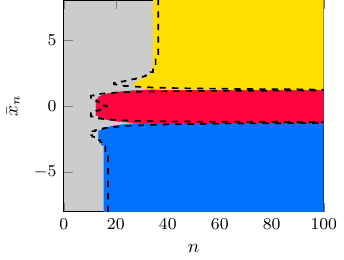}
    \caption{Shift-in-Mean scenario: Comparison of the \gls{ao} policy and the strictly optimal policy.}
    \label{fig:shiftInMean_policy}
\end{figure}

To validate the performance of the \gls{ao} scheme and to compare its performance to the performance of the competitors, a Monte Carlo simulation with $10^6$ runs is performed.
The results are summarized in \cref{tbl:shiftInMean}.
In \cref{tbl:shiftInMean_err}, the detection and estimation errors as well as the nominal levels are summarized.
One can directly see that the \gls{ao} procedure and the optimal procedure fulfill the constraints with equality, despite the Monte Carlo uncertainty during the design process.
For the S-TS procedure, however, the empirical detection errors are much smaller than the nominal ones, whereas the empirical estimation errors violate the constraints.
This is caused by the fact that the stopping time of the S-TS procedure is determined by the stopping rule of the \gls{msprt}, which does not take any uncertainty about the true parameter into account.
For the average run-lengths that are summarized in \cref{tbl:shiftInMean_rl}, one can see that the \gls{ao} and the optimal procedure use almost the same average number of samples.
Finally, for the FSS-TS, one can see that most of the constraints are fulfilled, but the constraints on $\detErr[2]$ and $\estErr[3]$ are violated.
This is caused by the fact that the FSS-TS neither controls the detection nor the estimation errors.

\cref{fig:shiftInMean_policy} shows the policy of the \gls{ao} and the optimal scheme, where the filled areas correspond to the \gls{ao} policy and the dashed line corresponds to the boundaries of the optimal policy.
That is, the boundary between the different regions in the state space, e.g., the region in which the procedure continues sampling and the one in which the procedure stops sampling and decides in favor of $\Hyp_1$.
Although the exact shapes of both policies differ, the general shape of both policies is similar and the \gls{ao} policy looks like a smoothed version of the optimal one.
Even though we made no restrictions on the maximum number of samples for the \gls{ao} policy, the maximum number of samples used by the \gls{ao} policy is limited.
It can be seen that the corridor between the regions for stopping and deciding in favor of $\Hyp_1/\Hyp_2$ closes at around $40$ samples and the corridor between the regions for stopping and deciding in favor of $\Hyp_2/\Hyp_3$ closes at around $32$ samples.
Contrary to this, the corridors in which the optimal scheme continues sampling still exist until the maximum number of samples is reached.
 \subsection{Joint Symbol Decoding and Noise Power Estimation}
\begin{figure}\centering
 \includegraphics{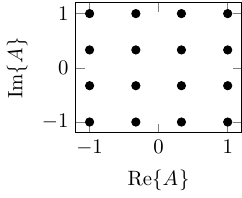}
 \caption{Constellation diagram of the 16-\gls{qam} symbols.}
 \label{fig:QAM_symbols}
\end{figure}
The second example illustrates how to apply the proposed theory to real-world problems for which designing a strictly optimal procedure becomes highly challenging.
In this example, the proposed theory is applied to the problem of joint symbol decoding and noise power estimation.
More precisely, we consider a $16$-\gls{qam} symbol which is transmitted over an \gls{awgn} channel with random noise power.
The aim is to simultaneously infer the transmitted symbol and the noise power using as few samples as possible.
Here, the inference of the transmitted symbol, i.e., the symbol decoding, is formulated as a hypothesis test.
The signal model for the complex-valued received signal is given by
$\obsSingleScalar_n = \QAMsymbol[] + \noiseSingle_n\,,$
where $\QAMsymbol[]\in\{\QAMsymbol[1],\ldots,\QAMsymbol[16]\}$ is the transmitted symbol and $\noiseSingle_n$ is the noise process.
The noise process is assumed to follow a circularly symmetric Gaussian distribution with zero-mean and variance $\sigma^2$.
Hence, the conditional distribution of the received signal is given by
\begin{align}
 \RVsingleScalar_n\given\Hyp_m,\noisePower \iid \complexGauss(\QAMsymbol, \noisePower)\,,
\end{align}
where $\complexGauss(\QAMsymbol, \noisePower)$ denotes the circularly symmetric Gaussian distribution with mean $\QAMsymbol$ and variance $\noisePower$.
The \gls{pdf} of the circularly symmetric Gaussian distribution is $p(\obsSingleScalar\given\Hyp_m,\noisePower) = \frac{1}{\pi\sigma^2}\exp\Bigl(-\frac{\lvert\obsSingleScalar - \QAMsymbol\rvert^2}{\sigma^2}\Bigr)\,$.
The noise power distribution is modelled as an inverse Gamma distribution.
Finally, the hypotheses can be formulated as
\begin{align*}
 \Hyp_m: & \quad \RVsingleScalar_n\given\noisePower \iid \complexGauss(\QAMsymbol, \noisePower)\,, \quad\noisePower\sim \invGam(a,b)\,,
\end{align*}
where $m\in\{1,\ldots,16\}$ and $\invGam(a,b)$ denotes the inverse Gamma distribution with shape and scale parameters $a$ and $b$, respectively.
The \gls{pdf} of the inverse Gamma distribution is given by \cite[Definition 8.22]{barber2012bayesian}
\begin{align*}
 p(\noisePower) = \frac{b^a}{\Gamma(a)}(\sigma^2)^{-a-1}e^{-\frac{b}{\sigma^2}}\,,
\end{align*}
where $\Gamma(\cdot)$ denotes the Gamma function.

In the supplement \cite[\cref{ext:app:QAM_assumptions}]{reinhard2022Supplement}, it is shown that \crefrange{ass:fishInf}{ass:klTaylor} hold.
Furthermore, if $a>2$, the prior distribution of the noise power has finite variance. Hence, the proposed policy is \gls{ao}.

One could, at least in theory, design an optimal
sequential procedure as the sufficient statistic $[\mathrm{Re}\{\bar{x}_n\}, \mathrm{Im}\{\bar{x}_n\}, n^{-1}\sum_{i=1}^n|x_i|^2]$ can serve as a low-dimensional representation of the data in the sense of \cite[Assumption A3]{reinhard2021multiple}.
However, although this sufficient statistic exists and seems to be relatively low-dimensional, designing the optimal scheme on the discretized three-dimensional state space is unlikely to be feasible in practice.
Therefore, we only present results for the \gls{ao} and the two-step procedures.

The following parameters are used in this simulation.
The shape and scale parameters of the prior distribution of the noise power are set to $2.1$ and $0.9$, respectively.
The constellation diagram of the \gls{qam} symbols is depicted in \cref{fig:QAM_symbols}.
Finally, the nominal detection and estimation error levels are set to $0.01$.
The gradients in \cref{alg:projQuasiNewton} are estimated by $10^6$ Monte Carlo runs.
Moreover, the initial set of cost coefficients and the tolerances are set to $\detCost[m]^{(0)}=\estCost[m]^{(0)}=500$, $\detTol=\estTol=0.005$ for all $m\in\{1,\ldots,16\}$, respectively.
\begin{figure}[tp]
 \centering
 \includegraphics{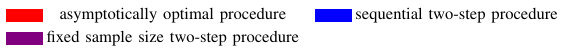}\\
 \vspace{-1em}
 \subfloat[Detection errors.\label{fig:QAM_det}]{\includegraphics{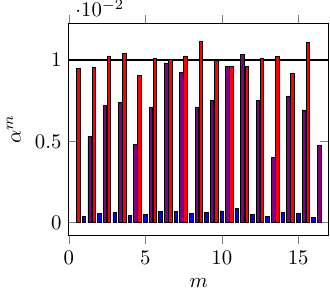}}\hfil
 \subfloat[Estimation errors.\label{fig:QAM_est}]{\includegraphics{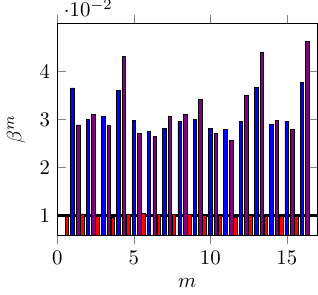}}
 \vspace{-.5em}
 \subfloat[Average run-lengths.\label{fig:QAM_rl}]{\includegraphics{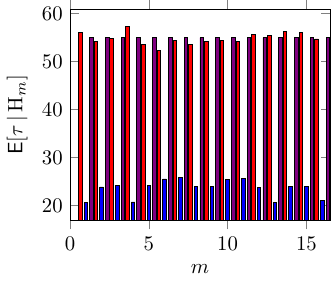}}
 \caption{Joint symbol decoding and noise power estimation: Simulation results.}
 \label{fig:QAM_results}
\end{figure}
To evaluate the performance of the designed procedure, $10^6$ Monte Carlo runs are performed with the \gls{ao} policy and the two-step procedures.
The results are summarized in \cref{fig:QAM_results}.
It can be seen from \cref{fig:QAM_det} and \cref{fig:QAM_est} that the \gls{ao} procedure hits the nominal detection and estimation error levels within the design tolerances.
Contrary to this, the detection errors of the S-TS procedure are much smaller than the tolerated $0.01$.
In addition, the estimation errors of the S-TS procedure procedure are far above the tolerated errors.
This is again due to the fact that the \gls{msprt} which defines the stopping time of the S-TS procedure only takes the uncertainty about the true hypothesis into account and does not consider the uncertainty about the noise power.
From the average run-lengths, which are depicted in \cref{fig:QAM_rl}, one can see that the S-TS procedure is much faster than the \gls{ao} one.
However, this comes at the cost of severely violating the estimation constraints.
For the FSS-TS one can see that almost all detection constraints are fulfilled.
The constraints on the estimation errors, however, are severely violated although the FSS-TS procedure uses a similar amount of data as the \gls{ao} procedure uses on average.

 \subsection{Estimating a Random Parameter in Unknown Noise}\label{sec:num_res_unknown_noise}
A similar version of this example, which is used to show the versatile applicability of the proposed framework, has in a similar version already been published in \cite[Section 4.4.3]{reinhard2021_diss}.
Designing optimal estimators is a common task in many signal processing scenarios.
To do so, the distribution of the stochastic process must be known exactly.
However, in many practical scenarios, the distribution of the noise process is not known exactly and using a wrong statistical model deteriorates the performance of the estimator.
This holds particularly for the sequential case because a wrong statistical model may cause the estimator to stop sampling too early.

In this example, we wish to estimate a random parameter in additive white noise.
The variance of the noise process is assumed to be known, but its distribution is subject to uncertainty.
More precisely, the noise either follows a Gaussian or a Laplacian distribution.
The different distributions can, for example, be caused by different environmental conditions.

A common way of approaching this problem is a two-step procedure that first identifies the noise process and subsequently applies an estimator that is optimal for the identified process.
However, when moving to the sequential case using such a procedure is not a good choice.
Although the noise process can be reliably detected using an appropriate sequential hypothesis test, it is not possible to control the estimation error levels, which are of primary interest in this case.

Besides the two-step approach that only allows to control the detection errors, the framework of sequential joint detection and estimation is well suited for approaching this problem.
Although one is primarily interested in an estimate of the parameter, the identification of the noise process, i.e., performing a hypothesis test, is an intermediate but necessary step.
The framework of sequential joint detection and estimation allows to control the estimation error levels as well as the detection error levels, i.e., the probability that the noise process is misidentified.

In this example, we consider a Gamma distributed random parameter that is embedded in additive white noise, which follows a Gaussian distribution under the null hypothesis and a Laplace distribution under the alternative.
More formally, the two hypotheses are defined as
\begin{align*}
 \Hyp_1: & \; \RVsingleScalar_n \given \paramScalar \sim \norm{\paramScalar}{\sigma^2}\,,\; \paramScalar\sim\Gam(a,b)\,,\\
 \Hyp_2: & \; \RVsingleScalar_n \given \paramScalar \sim \Lapl(\paramScalar,\varsigma)\,,\; \paramScalar\sim\Gam(a,b)\,,
\end{align*}
where $\Lapl(\paramScalar,\varsigma)$ denotes the Laplace distribution with location parameter $\paramScalar$ and scale parameter $\varsigma$.
The \gls{pdf} of the Laplace distribution is given by
\begin{align}
 p(\obsSingleScalar\given\paramScalar,\varsigma) = \frac{1}{2\varsigma}e^{-\frac{\lvert\obsSingleScalar-\paramScalar\rvert}{\varsigma}}\,.
\end{align}
For Gaussian random variables with unknown location parameter it is possible to find a sufficient statistic as shown in \cite{reinhard2021multiple}.
However, this is not possible for a Laplacian distributed random variable with unknown location parameter.
Since the existence of a sufficient statistic is a requirement for designing a strictly optimal procedure, it is not possible to design such a procedure for the problem at hand.

The parameters used in this example are as follows.
Both hypotheses have equal prior probability and the shape and scale parameters of the Gamma distribution are set to $2.1$ and $0.9$, respectively.
The Gaussian noise has a variance of $\sigma^2=4$ and the shape parameter of the Laplacian distribution is set to $\varsigma=\sqrt{0.5\sigma^2}$ so that both noise processes have the same variance.
Moreover, the error probabilities and the \gls{MSE} of the designed \gls{ao} procedure should not exceed $0.05$.

In \cite[\cref{ext:app:diffNoise_assumptions}]{reinhard2022Supplement}, it is shown that \crefrange{ass:fishInf}{ass:klTaylor} hold.
Furthermore, as $a>2$ holds, the prior distribution of the noise power has finite variance and, hence, the proposed policy is \gls{ao}.

During the design process, the gradients are estimated via $10^4$ Monte Carlo runs and the tolerances are set to $\detTol=\estTol=0.005$.
The \gls{ao} procedure and the two-step procedures are validated by Monte Carlo simulation with $10^4$ runs, whose results are summarized in \cref{tbl:ao_diffNoise}.

\begin{table}
    \centering
    \caption{Estimating a random signal under noise uncertainty: simulation results.}
    \label{tbl:ao_diffNoise}
    \subfloat[Detection and estimation errors.\label{tbl:diffNoise_err}]{\begin{tabular}{@{}rccccc@{}}
\toprule
& nominal & tolerance & AO & S-TS & FSS-TS\\
	\midrule
$\detErr[1]$ & $0.050$ & $\pm0.005$& $0.047$ & $0.035$ & $0.094$\\
	$\detErr[2]$ & $0.050$ & $\pm0.005$& $0.050$ & $0.043$ & $0.107$\\
	$\estErr[1]$ & $0.050$ & $\pm0.005$& $0.052$ & $0.090$ & $0.053$ \\
	$\estErr[2]$ & $0.050$ & $\pm0.005$& $0.054$ & $0.078$ & $0.030$ \\
\bottomrule
\end{tabular}
 }\hspace{1cm}
\subfloat[Average run-lengths.\label{tbl:diffNoise_rl}]{\begin{tabular}{@{}rccc@{}}
	\toprule
& \gls{ao} & S-TS & FSS-TS \\
	\midrule
	
	$\E[\tau\given\Hyp_1]$  & $72.27$ & $56.00$ & $61.00$\\
	$\E[\tau\given\Hyp_2]$  & $48.91$ & $50.66$ & $61.00$\\

	$\E[\tau]$ & $60.57$ & $53.32$ & $61.00$\\

	\bottomrule
\end{tabular}
 }
\end{table}
From \cref{tbl:diffNoise_err}, one can see that the \gls{ao} procedure hits the nominal detection and estimation error levels within the design tolerance.
Contrary to this, the empirical error probabilities of the S-TS procedure are smaller than the nominal levels.
However, the S-TS procedure violates the constraints on the estimation errors as its stopping rule does not take the uncertainty about the true parameter into account, although the estimation performance is of primary interest in this example.
This effect is particularly visible under $\Hyp_1$.
The FSS-TS on the other hand, shows a different performance.
The constraints on the detection error levels are severely violated, while the estimation errors are maintained.
Contrary to the sequential procedures which adapt the sample size to the certainty about the phenomenon if interest, the FSS-TS always uses a fixed number of samples.
Therefore, neither the detection nor the estimation error levels can be controlled directly and the error levels are purely determined by the number of samples.
For this example, one can see that the desired constraints on the estimation errors are fulfilled, while the constraints on the detection error are severely violated.

From \cref{tbl:diffNoise_rl}, which summarizes the average sample number of all procedures, one can see that the S-TS procedure uses, on average, significantly less samples than the \gls{ao} procedure.
However, this comes at the cost of violating the estimation constraints.
Especially under the null hypothesis, the difference in the average sample number is large, which is consistent with the fact that the violation of the estimation constraint by the S-TS procedure is most severe under this hypothesis.
Interestingly, under the alternative the \gls{ao} procedure uses on average even less samples than the S-TS procedure while resulting in more accurate estimates.
The FSS-TS procedure uses on average slightly more samples than the \gls{ao} procedure and, hence, it is the procedure that uses on average the most number of samples.

  \section{Conclusions}\label{sec:conclusions}
We have proposed an asymptotically optimal procedure for the problem of sequential joint detection and estimation.
This procedure has been designed such that it fulfills constrains on the detection and estimation errors and, for the case when the nominal detection and estimation errors tend to zero, it minimizes the expected number of used samples.
The proposed asymptotically optimal stopping rule has been obtained by thresholding the instantaneous cost function that is parameterized by a set of cost coefficients.
It has further been shown that similarly to the strictly optimal procedures, there exists a strong connection between the derivatives of the solution of the optimal stopping problem with respect to the cost coefficients and the corresponding performance measures.
By exploiting this connection, a projected quasi-Newton method to optimally choose these coefficients has been introduced.
To validate the proposed theory, three numerical examples have been conducted.
The first example has been used to illustrate the differences compared to the strictly optimal policy.
Moreover, the marginal increase in the expected number of used samples when using an asymptotically optimal procedure has been shown by this example.
The remaining examples have been used to prove the applicability to a scenario for which no low-dimensional representation of the data exists.
In both numerical examples, the proposed method meets the required detection and estimation performance, whereas a two-step procedure using a sequential detector followed by an \gls{mmse} estimator failed to meet the required estimation accuracy.

Future research will consider asymptotically optimal procedures for the general non-\gls{iid} case.
Decomposable graphical models \cite{chen2019testing, chen2021ordering} might be beneficial for addressing this problem.

\appendix
\crefalias{section}{appsec}

\bibliographystyle{elsarticle-num}
\bibliography{IEEEabrv,mrabbrev,references} \end{document}

\newcommand{\new}[1]{\textcolor{blue}{#1}}
\newcommand{\IEEEQED}{\mbox{\rule[0pt]{1.3ex}{1.3ex}}}

\date{}

\maketitle

\author{
	\begin{center}
\textbf{\large Dominik Reinhard$^1$, Michael Fau\ss{}$^2$, and Abdelhak M. Zoubir$^1$} \\
		$^1$ Signal Processing Group, Technische Universit\"at Darmstadt, 64283 Darmstadt, Germany \\
		$^2$ Department of Electrical Engineering, Princeton University, Princeton, NJ 08544, USA
	\end{center}
}
\symbolfootnote[0]{\normalsize Address correspondence to Dominik Reinhard,
	Signal Processing Group, Technische Universit\"at Darmstadt, Merckstra\ss{}e 25, 
	64283 Darmstadt, Germany; E-mail: reinhard@spg.tu-darmstadt.de
}

\appendix
\section{Parameters of the Projected Quasi-Newton Algorithms}
This section summarizes the choice of the parameters used for the projected quasi-Newton method.

When setting the parameters for the projected quasi-Newton method, the goal is to reduce the number of Monte Carlo runs while not degrading the performance of the line search too much.
For the step size selection presented in  this work, the step size reduction factor $\gamma_\text{red}$, the maximum and the minimum resulting step length $\gamma_\text{max}$ and $\gamma_\text{min}$ have to be set.
The step size reduction factor $\gamma_\text{red}$ is set to $0.5$ which turned out to be a good tradeoff between reducing the step size too fast and requiring too many Monte Carlo simulations.
A step that is too large and increases the coefficients $\detCost[m]$ and $\estCost[m]$, $m\in\{1,\ldots,M\}$ results in a very high average run-length in the next Monte Carlo simulation.
Although the results of this Monte Carlo simulation lead to a rejection of the candidate step size, it results in a waste of computational resources due to the potentially long execution time of the Monte Carlo simulation.
Therefore, the maximum step length is limited to $\gamma_\text{max}$.
For the simulations used in this work, it is set to $\gamma_\text{max}=100$.
Very small resulting step sizes usually need a lot of reduction steps, which are equally to a large number of Monte Carlo simulations, are ineffective and the change in the objective may be in the range of Monte Carlo uncertainty.
Therefore, very small steps are to be avoided and the minimum resulting step length is set to $\gamma_\text{min}$.
This means that the step size reduction is automatically stopped if $\|\gamma\invHess^{(k)}\hat\nabla^{(k)}\| \leq \gamma_\text{min}$.
For the simulations used in this work, it is set to $\gamma_\text{min}=1$.
 \section{Proof of Lemma \ref{ext:lem:convPostH}}\label{proof:convPostH}
We first consider the sequence of random variables
 \begin{align*}
  \frac{1}{n} \log p(\obsSeq\given\Hyp_m,\param_m) =
  \frac{1}{n} \sum_{i=1}^n \log p(\obsSingle_i\given\Hyp_m,\param_m)\,.
 \end{align*}

According to the strong law of large numbers, it holds that
 \begin{align}\label{eq:convLikelihood}
\frac{1}{n} \log p(\obsSeq\given&\Hyp_m,\param_m)  \xrightarrow{a.s.}  \E\bigl[\log p(\RVsingle\given\Hyp_m,\param_m)\given \trueHyp, \trueParam\bigr]\,.
\end{align}
The right hand side of \cref{eq:convLikelihood} can be written as
 \begin{align}
  \E\bigl[\log p(\RVsingle\given\Hyp_m,\param_m)\given \trueHyp, \trueParam\bigr] = - \KL{p(\RVsingle\given\trueHyp,\trueParam)}{\,p(\RVsingle\given\Hyp_m,\param_m)} - H^\star\,,
  \label{eq:kl_def}
 \end{align}
 where $H^\star$ denotes the entropy of $p(\RVsingle\given\trueHyp, \trueParam)$, that is,
 \begin{align*}
  H^\star = - \E\bigl[\log p(\RVsingle\given\trueHyp,\trueParam) \given \trueHyp, \trueParam\bigr].
 \end{align*}
 In what follows, we use the short hand notation
 \begin{align*}
 	\KLsh{\Hyp_m}{\param_m} \coloneqq \KL{p(\RVsingle\given\Hyp_{m^\star},\param_{m^\star}^\star)}{\,p(\RVsingle\given\Hyp_m,\param_m)},
 \end{align*}
 which emphasizes that the KL divergence in \eqref{eq:kl_def} is a function of the hypothesis $\Hyp_m$ and the parameter $\param_m$. By \cref{ext:ass:posKL}, it holds that
 \begin{align}
 	\KLsh{\Hyp_m}{\param_m} \begin{cases}
 														= 0, & (\Hyp_m, \param_m) = (\trueHyp, \trueParam) \\
 													  > 0, & \text{otherwise}
 													\end{cases}
 	\label{eq:kl_cases}
 \end{align}
 That is, the KL divergence in \eqref{eq:kl_def} is zero if and only if the assumed sampling distribution is identical to the true sampling distribution, which in turn implies that the assumed hypothesis and parameter need to match the true ones. In the large sample regime, we can use \eqref{eq:convLikelihood} to express $p(\obsSeq\given\Hyp_m,\param_m)$ as a function of $\KLsh{\Hyp_m}{\param_m}$, namely,
 \begin{align}
  p(\obsSeq\given\Hyp_m,\param_m) \xrightarrow{a.s.} \exp\bigl(- n\KLsh{\Hyp_m}{\param_m} - nH^\star \bigr)
  \label{eq:px_convergence}
 \end{align}
 for $n\rightarrow\infty$.

Now, the joint posterior probability of hypothesis $\Hyp_m$ and parameter $\param_m$ is given by
 \begin{align}
  p(\Hyp_m, \param_m \given \obsSeq) & = \frac{p(\obsSeq\given\Hyp_m, \param_m)p(\Hyp_m, \param_m)}{p(\obsSeq)}\,, \nonumber\\
                                 & = \frac{p(\obsSeq\given\Hyp_m, \param_m)p(\Hyp_m, \param_m)}{\sum\limits_{k=1}^M \int p(\obsSeq\given\Hyp_k,\param_k)p(\Hyp_k, \param_k)\dInt\param_k p(\Hyp_k)}\,. \label{eq:posterior}
 \end{align}
 Taking the limit $n\rightarrow\infty$ and using \eqref{eq:px_convergence} we obtain
 \begin{align*}
 	p(\Hyp_m, \param_m \given \obsSeq) \xrightarrow{a.s.} \frac{\exp\bigl(- n\KLsh{\Hyp_m}{\param_m} \bigr)p(\Hyp_m, \param_m)}{\sum_{k=1}^M \int\exp\left(- n\KLsh{\Hyp_k}{\param_k} \right)p(\Hyp_k, \param_k)\dInt\param_k}
\end{align*}
 From \eqref{eq:kl_cases} it follows that
 \begin{align}
 	\exp\bigl(- n\KLsh{\Hyp_m}{\param_m} \bigr) \xrightarrow{a.s.}
 	\begin{cases}
		1, & (\Hyp_m, \param_m) = (\trueHyp, \trueParam) \\
		0, & \text{otherwise}
	\end{cases}
 \end{align}
This in turn implies
\begin{align*}
	\exp\bigl(- n\KLsh{\Hyp_m}{\param_m} \bigr)p(\Hyp_m, \param_m) \xrightarrow{a.s.}
	\begin{cases}
		p(\trueHyp, \trueParam), & (\Hyp_m, \param_m) = (\trueHyp, \trueParam) \\
		0, & \text{otherwise}
	\end{cases}
\end{align*}
Since the denominator in \eqref{eq:posterior} is merely a normalization factor, it holds that
\begin{align*}
	p(\Hyp_m, \param_m \given \obsSeq) \propto \exp\bigl(- n\KLsh{\Hyp_m}{\param_m} \bigr)p(\Hyp_m, \param_m)
\end{align*}
as $n \to \infty$. This implies that $p(\Hyp_m, \param_m \given \obsSeq)$ tends to zero almost everywhere, except for the point $(\trueHyp, \trueParam)$, at which it needs to tend to infinity in order to remain a valid density function. In other words, $p(\Hyp_m, \param_m \given \obsSeq)$ converges to a point mass at $(\trueHyp, \trueParam)$:
\begin{align*}
	p(\Hyp_m, \param_m \given \obsSeq) \xrightarrow{a.s.} \delta_{(\trueHyp, \trueParam)}(\Hyp_m, \param_m),
\end{align*}
where $\delta_x$ denotes Dirac's delta distribution at $x$. Finally, marginalizing over $\param_m$ yields
\begin{align*}
	p(\Hyp_m\given \obsSeq) = \int p(\Hyp_m, \param_m \given \obsSeq) \dInt\param_m \xrightarrow{a.s.} \int \delta_{(\trueHyp, \trueParam)}(\Hyp_m, \param_m) \dInt\param_m =
	\begin{cases}
		1 & \Hyp_m = \trueHyp\,,\\
		0 & \Hyp_m \neq \trueHyp\,.
	\end{cases}
\end{align*}

This concludes the proof.\hfill\IEEEQED

\section{Proof of Lemma \ref{ext:lem:convPostVar}}\label{proof:convPostVar}

The proof of \cref{ext:lem:convPostVar} follows roughly from the Bernstein-von-Mises theorem.
However, there exist different statements of the Bernstein-von-Mises theorem that are based on different, sometimes very technical, assumptions.
In this proof, we apply a misspecified version of the Bernstein-von-Mises theorem \cite{kleijn2012bernstein} since it is based on conditions that are easily verifiable.
Let $p_m^\star$ be the \gls{pdf} of the true sampling distribution.
Under the assumption that the sampling distribution is not necessarily part of the assumed model, it is stated in \cite{kleijn2012bernstein} that under certain conditions that will be explained shortly, it holds that
\begin{align*}
 \Bigl\lvert p(\param_m\given\Hyp_m,\obsSeq) - \mathcal{N}\bigl(\est{m}, (n\boldsymbol V_{\pseudoTrueParam})^{-1}\bigr)\Bigr\rvert \rightarrow 0\,,
\end{align*}
where $\est{m}$ is some suitable estimator, $\pseudoTrueParam$ is the parameter that minimizes the \gls{kl}
between $p_m^\star(\obsSingle)$ and $p(\obsSingle\given\Hyp_m,\param_m)$
and
 \begin{align} \label{eq:minSecDerKL}
  \bigl[\boldsymbol V_{\pseudoTrueParam}\bigr]_{i,j} = - \frac{\partial^2}{\partial  \paramElement_m^i \partial \paramElement_m^j} \KL{ p_m^\star(\obsSingle) }{ p(\obsSingle\given\Hyp_m,\param_m )}\biggr\rvert_{\param_m = \tilde\param^\star}\,.
 \end{align}
 In the previous equation $[\bm A]_{i,j}$ denotes the $i$th row and the $j$th column of matrix $\bm A$ and $\theta_m^i$ denotes the $i$th element of the parameter vector $\bm \theta_m$.
Assuming that there is no model mismatch, i.e., $\Hyp_m = \trueHyp$, it follows that $\pseudoTrueParam=\trueParam$.
It can be shown that in this case \cref{eq:minSecDerKL} becomes \gls{fim}, i.e.,
\begin{align*}
 \boldsymbol V_{\pseudoTrueParam} = \fishInfMtx{m^\star}\,.
\end{align*}
According to \cite[Lemma 2.1.]{kleijn2012bernstein}, the posterior distribution of $\paramRV_m$ converges in total variation to a normal distribution with covariance matrix $n^{-1}\fishInfMtxInv{m^\star}$ if \cref{ext:ass:localLipschitz} and \cref{ext:ass:klTaylor} hold.
Therefore, we can conclude that $\Tr(\errorCovMatrix[\trueHypIdx])\xrightarrow{a.s.} 0$ and $n\Tr(\errorCovMatrix[\trueHypIdx])\xrightarrow{a.s.} \Tr(\fishInfMtxInv{m^\star})$.

\section{Proof of Lemma \ref{ext:lem:finiteEx}}\label{proof:finiteEx}
From the definition of $\gDiv(\obsSeq)$ it follows that
\begin{align*}
 \E[\gDiv(\obsSeq)] \leq \E[\DDiv(\obsSeq)]\,,
\end{align*}
where the right hand side of the inequality is equal to
\begin{align}\label{eq:expValue}
  & \sum_{i=1, i\neq m}^M \detCostDiv \E[p(\Hyp_i\given\obsSeq)] + \estCostDiv[m] \E[p(\Hyp_m\given\obsSeq)\Tr(\errorCovMatrix)]\,.
\end{align}
For the first term, it holds that
\begin{align*}
 \E[p(\Hyp_i\given\obsSeq)] = p(\Hyp_i) < \infty\,,\quad\forall i\in\{1,\ldots,M\},\; n\geq1\,.
\end{align*}
The expectation in the second term of \cref{eq:expValue} can be rewritten as
\begin{align}\label{eq:var_exp}
\sum_{k=1}^{K_m}\int \Var[\paramRVElement_m^k\given\Hyp_m,\obsSeq] p(\Hyp_m,\obsSeq)\dInt\obsSeq\,.
\end{align}
By using the definition of the posterior variance and expanding the square, the above integral becomes
\begin{align}\label{eq:var_expr_full}
& \E\Bigl[\bigl(\paramRVElement^k_m\bigr)^2\given\Hyp_m\Bigr] - 2\Bigl(\E[\paramRVElement^k_m\given\Hyp_m]\Bigr)^2  + \int \Bigl(\E[\paramRVElement^k_m\given\Hyp_m,\obsSeq]\Bigr)^2p(\obsSeq\given\Hyp_m)\dInt\obsSeq\,.
\end{align}
Due to Jensen's inequality \cite{everitt2010cambridge}, it holds that
\begin{align*}
 \Bigl(\E[\paramRVElement^k_m\given\Hyp_m,\obsSeq]\Bigr)^2 \leq \E\Bigl[\bigl(\paramRVElement^k_m\bigr)^2\given\Hyp_m,\obsSeq\Bigr]\,.
\end{align*}
Therefore, the last term in \cref{eq:var_expr_full} is upper bounded by
\begin{align*}
 \int\!\!\E\Bigl[\bigl(\paramRVElement_m^k\bigr)^2\given\Hyp_m,\obsSeq\Bigr] p(\obsSeq\given\Hyp_m)\dInt\obsSeq \!=\! \E\Bigl[\bigl(\paramRVElement_m^k\bigr)^2\given\Hyp_m\Bigr]\,.
\end{align*}
Hence, we can state that
\begin{align*}
    2\E\Bigl[\bigl(\paramRVElement_m^k\bigr)^2\given\Hyp_m\Bigr] - 2\Bigl(\E[\paramRVElement_m^k\given\Hyp_m]\Bigr)^2 = 2\Var[\paramRVElement_m^k\given\Hyp_m]
\end{align*}
is an upper bound for the integral in \cref{eq:var_exp}.
It follows that
\begin{align*}
 \E[p(\Hyp_m\given\obsSeq)\Tr(\errorCovMatrix)] \leq 2\sum_{k=1}^{K_m} \Var[\paramRVElement_m^k\given\Hyp_m]\,,
\end{align*}
which is finite according to \cref{ext:ass:finiteSecondOrder} for all $m\in\{1,\ldots,M\}$ and, hence, \cref{eq:var_exp} is finite.
Finally, as the coefficients $\detCostDiv[m]$, $\estCostDiv[m]$ are finite for all $m\in\{1,\ldots,M\}$, we conclude that the expectation of $\DDiv(\obsSeq)$, and, hence, also the expectation of $\gDiv(\obsSeq)$ is finite.

\section{Proof of Theorem \ref{ext:theo:derivatives}}\label{app:proofDerivatives}
Assume that both schemes, the optimal one presented in \cite{reinhard2019bayesian,reinhard2021multiple} and the \gls{ao} scheme, are truncated, i.e., the number of samples is restricted not to exceed $N$.
Then, a scheme with stopping rule $\stopR[]$ can be characterized by the set of cost functions
\begin{align}
 \begin{split}\label{eq:defCostFct}
  V_n(\obsSeq;\stopR[]) &= \stopR(n+g(\obsSeq))  + (1-\stopR)\E[V_{n+1}(\RVseq[n+1];\stopR[])\given\obsSeq]\,,\\
  V_N(\obsSeq[N];\stopR[]) &=  N+g(\obsSeq[N])\,.
 \end{split}
\end{align}
The cost function $V_n(\obsSeq;\stopR[])$ describes the cost of the optimal stopping problem when using the stopping rule $\stopR[]$ after observing $\obsSeq$.
Therefore, it holds that $V_0(\stopR[]) = \E[\tau + g(\obsSeq[\tau]);\stopR]$.
Moreover, the definition of the \gls{ao} stopping rule in \cref{ext:eq:AO_sjde} implies
\begin{align}\label{eq:optStopConv}
 \E[\aoRL + g(\obsSeq[\aoRL])] \rightarrow \E[\optRL + g(\obsSeq[\optRL])]
\end{align}
as $\max\{\detConstr_1,\ldots,\detConstr_M, \estConstr_1,\ldots,\estConstr_M\}\; \rightarrow0$.
In \cref{eq:optStopConv}, $\optRL$ and $\aoRL$ denote the stopping time of the optimal and the \gls{ao} scheme, respectively.
Define the sequence of non-negative functions $\Delta V_n(\obsSeq) = V_n(\obsSeq;\stopRao[]) - V_n(\obsSeq;\stopRopt[])$, which can be written as
\begin{align}\label{eq:delta_Vn}
\begin{split}
 \Delta V_{n} = (\stopRao[n] - \stopRopt[n])\bigl(g(\obsSeq[n])+n\bigr) & + (\stopRopt[n]-\stopRao[n]) \E[V_{n+1}(\RVseq[n+1];\stopRopt[])\given\obsSeq[n]] \\
                & + (1-\stopRao[n])\E[\Delta V_{n+1}\given\obsSeq[n]]\,.
\end{split}
\end{align}
The convergence stated in \cref{eq:optStopConv} implies that $\Delta V_{0}\rightarrow0$.
For an arbitrary $n$, the first two terms in \cref{eq:delta_Vn} vanish if the stopping rule at time $n$ converges, whereas the last term vanishes if and only if $\Delta V_{n+1}\rightarrow 0$ almost everywhere.
The latter implies the convergence of the stopping rule and the cost functions at time $n+1$, which in turn only holds if the stopping rules at time $n+1$ converge, i.e., $\stopRao[n+1]-\stopRopt[n+1]\rightarrow0$, and the cost functions at time $n+2$ converge almost everywhere.
Finally, the stopping rules at time instances $n,\ldots,N$ have to convergence so that $\Delta V_n\rightarrow0$ holds almost everywhere.
In particular, \cref{eq:optStopConv} implies that the stopping rules convergence for all $n\geq 0$.

For an arbitrary stopping rule $\stopR[]$, the gradient of the cost function can, similarly to \cite[Lemma IV.2.]{reinhard2021multiple}, be written as
\begin{align*}
\begin{split}
 \frac{\partial}{\partial\detCost[m]}  V_n(\obsSeq;\stopR)
 & = \stopR \ind{\decR\neq m} p(\Hyp_m\given\obsSeq) \\
                                                          & \phantom{{}={}} + (1-\stopR)\E\biggl[\frac{\partial}{\partial\detCost[m]} V_{n+1}(\RVseq[n+1];\stopR)\bbgiven\obsSeq\biggr]\,,
\end{split}\\
 \frac{\partial}{\partial\detCost[m]}  V_N(\obsSeq[N];\stopR) & = \ind{\decR\neq m}p(\Hyp_m\given\obsSeq[N])\,,
\end{align*}
\begin{align*}
\begin{split}
 \frac{\partial}{\partial\estCost[m]}  V_n(\obsSeq;\stopR)
 & = \stopR \ind{\decR = m} p(\Hyp_m\given\obsSeq)\Tr(\errorCovMatrix) \\
                                                          & \phantom{{}={}} + (1-\stopR)\E\biggl[\frac{\partial}{\partial\estCost[m]} V_{n+1}(\RVseq[n+1];\stopR)\bbgiven\obsSeq\biggr]\,,
\end{split}\\
\begin{split}
\frac{\partial}{\partial\estCost[m]}  V_N(\obsSeq[N];\stopR) & = \ind{\decR = m} p(\Hyp_m\given\obsSeq[N])\Tr(\errorCovMatrix)\,.
\end{split}
\end{align*}
Using similar arguments as for the proof above, it can be shown that the gradient convergences if and only if the stopping rules converge.
As it was shown previously that the stopping rules convergence almost everywhere, also the gradients converge.

According to \cite[Theorem IV.1.]{reinhard2021multiple}, it holds that
\begin{align*}
 \frac{\partial}{\partial\detCost[m]} \E[V_1(\RVsingle_1;\stopRopt[])] & = p(\Hyp_m)\detErr\,,\\
 \frac{\partial}{\partial\estCost[m]} \E[V_1(\RVsingle_1;\stopRopt[])] & = p(\Hyp_m)\estErr\,,
\end{align*}
and as the gradients converge, this relationship holds also for the \gls{ao} stopping rule in the asymptotic case.
In this proof, a truncated version of the strictly optimal and the \gls{ao} procedure were considered.
As there is no restriction on the maximum number of samples for the \gls{ao} procedure, the limit $N\rightarrow\infty$ has to be taken.
However, this does not affect the above derivations.
This completes the proof.

\section{Proof of Theorem \ref{ext:theo:asymptotic_duality}} \label{proof:asymptotic_duality}
 For sufficiently small nominal error levels, i.e., such that \cref{ext:theo:derivatives} holds, the gradients of the objective in \cref{ext:eq:finalMaxProblem} are given by
\begin{align}\label{eq:grad}
\begin{split}
 \gradDet & = [\,p(\Hyp_1)(\detErr[1]-\detConstr_1)\,,\,\ldots\,,\,p(\Hyp_M)(\detErr[M]-\detConstr_M)\,]^\top\,,\\
 \gradEst & = [\,p(\Hyp_1)(\estErr[1]-\estConstr_1)\,,\,\ldots\,,\,p(\Hyp_M)(\estErr[M]-\estConstr_M)\,]^\top\,.
 \end{split}
\end{align}
As it is assumed that $\detCost[m],\estCost[m]$ are strictly positive, \cref{ext:eq:finalMaxProblem} attains its maximum only if $
 \gradDet = \gradEst = 0\,,$
which in turn only holds if
\begin{align*}
 \detErr[m] & = \detConstr_m\,,\quad \estErr[m]  = \estConstr_m\,,\quad m\in\{1,\ldots,M\}\,.
\end{align*}
Since the objective of \cref{ext:eq:finalMaxProblem} is equivalent to
$\E[\aoRL] + \sum_{m=1}^M p(\Hyp_m)(\detCost[m]\detErr + \estCost[m]\estErr)$, it holds that
\begin{align*}
 & \E[\aoRL + g(\obsSeq[\aoRL])] - \sum_{m=1}^M p(\Hyp_m)(\detCost[m]\detConstr_m + \estCost[m]\estConstr_m) \\
 = & \E[\aoRL] + \sum_{m=1}^M p(\Hyp_m)(\detCost[m]\detErr + \estCost[m]\estErr)  - \sum_{m=1}^M p(\Hyp_m)(\detCost[m]\detConstr_m + \estCost[m]\estConstr_m) \\
 = & \E[\aoRL]\,.
\end{align*}
According to \cref{ext:eq:AO_sjde}, we can conclude that
\begin{align*}
 \E[\aoRL] \rightarrow \E[\optRL]\,.
\end{align*}
Hence, the procedure that is parameterized by the solution of \cref{ext:eq:finalMaxProblem} fulfills all constraints with equality and asymptotically uses on average as few samples as possible, i.e., it solves \cref{ext:eq:constrProblem} asymptotically.

 \section{Proof of Assumptions \ref{ext:ass:fishInf} to \ref{ext:ass:klTaylor} for the Shift-in-Mean Scenario}\label{app:sim_assumptions}

Fisher's information is given by
\begin{align*}
 \fishInf[\paramScalar_m]{m} = \bigl(\sigma^2\bigr)^{-1}\,,
\end{align*}
for which $0<\fishInf[\paramScalar_m]{m}<\infty$ holds with probability one.
Hence, \cref{ext:ass:fishInf} is fulfilled.

The \gls{kl} is calculated as
\begin{align*}
 \KL{p(\obsSingleScalar\given\Hyp_m,\paramScalar_m)}{p(\obsSingleScalar\given\Hyp_k,\paramScalar_k)} = \frac{(\paramScalar_m-\paramScalar_k)^2}{2\sigma^2}\,.
\end{align*}
As the prior distributions have a disjoint support, the \gls{kl} is positive whenever $k\neq m$. Hence, \cref{ext:ass:posKL} holds.
Moreover, the \gls{kl} has a quadratic form, which means that it has a second order Taylor-expansion in the sense of \cref{ext:ass:klTaylor}.

It is left to show that \cref{ext:ass:localLipschitz} holds.
The log-likelihood function is given by
\begin{align}
 \log p(\obsSingleScalar\given\trueHyp,\paramScalar) = -0.5\log(2\pi\sigma^2) - \frac{(\obsSingleScalar - \paramScalar)^2}{2\sigma^2}\,.
\end{align}
It suffices to show that the absolute value of the first derivative of the log-likelihood function is bounded in a neighborhood of $\trueParamScalar$.
The score, i.e., the first derivative of the log likelihood function, is
\begin{align*}
 \frac{\partial}{\partial \paramScalar} \log p(\obsSingleScalar\given\trueHyp,\paramScalar) = \frac{\obsSingleScalar - \paramScalar}{\sigma^2}\,.
\end{align*}
Let $U$ denote the neighborhood of the true parameter $\trueParamScalar$, then for all $\paramScalar\in U$ the score is continuous and bounded.
Hence, by setting $f_{\trueParamScalar}(\obsSingleScalar) = \max_{\paramScalar\in U} \lvert \frac{\obsSingleScalar - \paramScalar}{\sigma^2} \rvert$
\cref{ext:ass:localLipschitz} is fulfilled.

\section{Proof of Assumptions \ref{ext:ass:fishInf} to \ref{ext:ass:klTaylor} for the QAM Scenario}\label{app:QAM_assumptions}
To show that \cref{ext:ass:fishInf} holds, the second derivative of the log-likelihood function has to be calculated, i.e., 
\begin{align*}
 \frac{\partial^2}{\partial (\noisePower)^2} \log p(\obsSingleScalar\given\Hyp_m,\noisePower) = \bigl(\noisePower\bigr)^{-2} - 2 \frac{\lvert\obsSingleScalar-\QAMsymbol\rvert^2}{\bigl(\noisePower\bigr)^3}\,.
\end{align*}
Taking the conditional negative expectation of the second-order derivative, yields Fisher's information, i.e.,
\begin{align*}
 \fishInf[\noisePower]{m} = \bigl(\noisePower\bigr)^{-2}\,.
\end{align*}
For Fisher's information it holds that $0<\fishInf[\noisePower]{m}<\infty$ except for a $P$-null set, i.e., it holds with probability one.
Therefore, \cref{ext:ass:fishInf} is fulfilled.

To prove \cref{ext:ass:posKL}, the \gls{kl} $\KL{p(\obsSingleScalar\given\Hyp_m,\noisePower)}{p(\obsSingleScalar\given\Hyp_k,\noisePower)}$ has to be calculated, which is given by
\begin{align*}
 \KL{p(\obsSingleScalar\given\Hyp_m,\noisePower)}{p(\obsSingleScalar\given\Hyp_k,\noisePower)} = \frac{\lvert\QAMsymbol[m]-\QAMsymbol[k]\rvert^2}{\sigma^2}\,.
\end{align*}
As the means $\QAMsymbol[m]$ and $\QAMsymbol[k]$ are deterministic and not equal, \cref{ext:ass:posKL} holds with probability one.
Contrary to the shift-in-mean scenario, the \gls{kl} $\KL{p(\obsSingleScalar\given\Hyp_m,\noisePower)}{p(\obsSingleScalar\given\Hyp_m,\noisePowerTilde)}$ has to be considered.
This is given by
\begin{align*}
 \KL{p(\obsSingleScalar\given\Hyp_m,\noisePower)}{p(\obsSingleScalar\given\Hyp_m,\noisePowerTilde)} = \log\biggl(\frac{\noisePowerTilde}{\noisePower}\biggr) + \frac{\noisePower}{\noisePowerTilde} - 1\,.
\end{align*}
It can be shown that the second order Taylor-expansion of the \gls{kl} is equal to
\begin{align*}
 \bigl(2\bigl(\noisePower\bigr)^2\bigr)^{-1}\bigl(\noisePowerTilde-\noisePower)^2 + o(|\noisePower-\noisePowerTilde|)\,,
\end{align*}
and hence, \cref{ext:ass:klTaylor} is fulfilled.

To show that \cref{ext:ass:localLipschitz} holds, the log-likelihood function
\begin{align}
 \log p(\obsSingleScalar\given\trueHyp,\noisePower) = -\log(\pi\noisePower) - \frac{|\obsSingleScalar - \QAMsymbol[\trueHypIdx]|^2}{\noisePower}
\end{align}
is calculated.
Taking the first derivative, yields
\begin{align*}
 \frac{\partial}{\partial \noisePower} \log p(\obsSingleScalar\given\trueHyp,\noisePower) = -\frac{1}{\noisePower} + \frac{|\obsSingleScalar - \QAMsymbol[\trueHypIdx]|^2}{\bigl(\noisePower\bigr)^2}\,.
\end{align*}
The derivative is continuous and bounded on an interval excluding zero.
In order to fulfill assumption \cref{ext:ass:localLipschitz}, the derivative must be bounded in a neighborhood $U$ around the true variance almost surely.
As $\noisePower=0$ occurs with probability zero and also $0\in U$ occurs with probability zero, this condition holds almost surely.
The function $f_{\trueParamScalar}(\obsSingleScalar)$ can then be set to $\max_{\noisePower\in U} \left\lvert \frac{1}{\noisePower} + |\obsSingleScalar - \QAMsymbol[\trueHypIdx]|^2\bigl(\noisePower\bigr)^{-2} \right\rvert$.

\section{Proof of Assumptions \ref{ext:ass:fishInf} to \ref{ext:ass:klTaylor} for Estimating a Random Parameter in Unknown Noise}\label{app:diffNoise_assumptions}
Under both hypotheses, Fisher's information can be calculated as
\begin{align}
 \fishInf[\paramScalar]{0} & = \frac{1}{\sigma^2} \quad \text{and} \quad \fishInf[\paramScalar]{1} = \frac{1}{\varsigma^2}\,,
\end{align}
which is positive and finite for $0<\sigma^2<\infty$ and $0<\varsigma^2<\infty$, respectively.

In order to prove that \cref{ext:ass:posKL} holds, the \gls{kl}
\begin{align}
 \KL{p(\obsSingle\given\Hyp_0,\paramScalar_0)}{p(\obsSingle\given\Hyp_1,\paramScalar_1)}
 = & -\log(2\varsigma) - 1 + \frac{1}{2}\log(2\pi\sigma^2) + \frac{2\varsigma^2 + (\paramScalar_0 - \paramScalar_1)^2}{2\sigma^2}
\end{align}
has to be calculated.
By setting $\varsigma=\sqrt{0.5\sigma^2}$, the \gls{kl} becomes
\begin{align}
 \KL{p(\obsSingle\given\Hyp_0,\paramScalar_0)}{p(\obsSingle\given\Hyp_1,\paramScalar_1)}
 =
 \log(\pi) - 1 + \frac{\sigma^2 + (\paramScalar_0 - \paramScalar_1)^2}{2\sigma^2}\,.
\end{align}
The last term is positive for all $\paramScalar_0$ and all $\paramScalar_1$ as it only consists of quadratic terms.
It further holds that $\log(\pi) - 1>0$ and, hence, the \gls{kl} is positive for all $\paramScalar_0$ and $\paramScalar_1$.

\cref{ext:ass:klTaylor} can be proven in close analogy to \cref{app:sim_assumptions} under $\Hyp_0$.
Under the alternative, however, the \gls{kl} calculates as
\begin{align}
 & \KL{p(\obsSingle\given\Hyp_1,\paramScalar_1)}{p(\obsSingle\given\Hyp_1,\paramScalar_1^\prime)} = e^{-\frac{|\paramScalar_1^\prime-\paramScalar_1|}{\varsigma}} + \frac{|\paramScalar_1^\prime-\paramScalar_1|}{\varsigma} - 1\,,
\end{align}
which has the second order Taylor-expansion
\begin{align}
 \frac{1}{2\varsigma} \lvert \paramScalar_1^\prime - \paramScalar_1\rvert^2\,.
\end{align}
Hence, \cref{ext:ass:klTaylor} holds under both hypotheses.

The proof that \cref{ext:ass:localLipschitz} holds for $\Hyp_0$ can be done similarly to the one outlined in \cref{app:sim_assumptions}.
Under $\Hyp_1$, the log likelihood function calculates as
\begin{align}
 \log p(\obsSingle\given\Hyp_1,\paramScalar) = - \log(2\varsigma) -  \frac{\lvert \obsSingleScalar-\paramScalar \rvert}{\varsigma}\,.
\end{align}
It further holds that
\begin{align}
 \Bigl\lvert \log \frac{p(\obsSingle\given\Hyp_1,\paramScalar)}{p(\obsSingle\given\Hyp_1,\paramScalar^\prime)}\Bigr\rvert = \Bigl\lvert - \frac{\lvert\obsSingleScalar-\paramScalar\rvert}{\varsigma} + \frac{\lvert\obsSingleScalar-\paramScalar^\prime\rvert}{\varsigma} \Bigr\lvert
 \leq \frac{1}{\varsigma} \lvert \paramScalar - \paramScalar^\prime\rvert\,,
\end{align}
and, hence, \cref{ext:ass:localLipschitz} holds.

\bibliographystyle{IEEEtran}
\bibliography{./IEEEabrv,./mrabbrev,./references}